\documentclass{LMCS}

\usepackage{enumerate}
\usepackage{hyperref}
\usepackage{amssymb,amsmath}
\usepackage{bussproofs}

\usepackage[show]{ed}

\newcommand{\Log}{\mathsf{L}}
\newcommand{\MPEM}{\mathsf{MPCEM}}
\newcommand{\N}{\mathsf{N}}
\newcommand{\D}{\mathsf{D}}
\newcommand{\R}{\mathsf{R}}
\newcommand{\Four}{\mathsf{4}}
\newcommand{\size}{\mathsf{size}}
\newcommand{\ID}{\mathsf{ID}}
\newcommand{\CEM}{\mathsf{CEM}}
\newcommand{\MP}{\mathsf{MP}}
\newcommand{\Thm}{\mathsf{T}}
\renewcommand{\Form}{\mathcal{F}}
\newcommand{\Land}{\bigwedge}
\newcommand{\Lor}{\bigvee}
\newcommand{\Rules}{\mathsf{R}}
\newcommand{\entails}{\vdash}
\newcommand{\hearts}{\heartsuit}

\newcommand{\Hilb}{\mathsf{H}}

\newcommand{\PL}{\mathsf{PL}}
\newcommand{\otto}{\leftrightarrow}
\newcommand{\Gen}{\mathsf{G}}
\newcommand{\Seq}{\mathsf{S}}

\newcommand{\Cut}{\mathsf{Cut}}

\newcommand{\Xtra}{\mathsf{X}}

\newcommand{\Struct}{\mathsf{A}}
\newcommand{\Pow}{\mathcal{P}}
\newcommand{\modimpl}{\rightarrow}
\newcommand{\lrule}[3]{(#1)\;\;\infrule{#2}{#3}}
\newcommand{\infrule}[2]{\frac{#1}{#2}}
\newcommand{\Krip}{\mathsf{K}}
\newcommand{\Kfour}{\mathsf{K4}}
\newcommand{\CK}{\mathsf{CK}}
\newcommand{\CKMP}{\mathsf{CKMP}}
\newcommand{\CKCEM}{\mathsf{CKCEM}}
\newcommand{\CKMPCEM}{\mathsf{CKMPCEM}}

\newcommand{\CKID}{\mathsf{CKID}}

\newcommand{\K}{\mathsf{K}}
\newcommand{\T}{\mathsf{T}}

\newcommand{\cto}{\Rightarrow}

\newenvironment{defn}{\begin{defi}}{\end{defi}}
\newenvironment{ex}{\begin{exa}}{\end{exa}}
\newenvironment{lemma}{\begin{lem}}{\end{lem}}
\newenvironment{propn}{\begin{prop}}{\end{prop}}
\theoremstyle{plain}\newtheorem{notn}[thm]{Notation}
\newenvironment{axarray}%
	       {\begin{array}{@{\hspace{3em}}p{3em}p{50em}}}{\end{array}}
\newcounter{blubber}

\newenvironment{sparenumerate}
{\begin{list}{$\bullet$}{
    \setlength{\leftmargin}{0pt}
    \setlength{\parsep}{0pt}
    \setlength{\itemindent}{4ex}
    \setlength{\itemsep}{0pt}
  }
}{\end{list}}

\def\doi{7 (1:4) 2011}
\lmcsheading%
{\doi}
{1--28}
{}
{}
{Jun.~24, 2010}
{Mar.~17, 2011}
{}

\begin{document}

\title[Generic Modal Cut Elimination]{Generic Modal Cut Elimination Applied to
Conditional Logics}

\author[D.~Pattinson]{Dirk Pattinson\rsuper a}	
\address{{\lsuper a}Department of Computing, Imperial College London}	
\thanks{{\lsuper a}Partially supported by EPSRC grant EP/F031173/1}	

\author[l.~Schr{\"o}der]{Lutz Schr{\"o}der\rsuper b}	
\address{{\lsuper b}DFKI Bremen and Department of Computer Science,  Universit\"at  Bremen}
\thanks{{\lsuper b}Work performed as part of the DFG project SCHR 1118/5-1}	

\keywords{Modal Logic, Proof Theory, Cut Elimination, Conditional
Logic}
\subjclass{F.4.1, I.2.3}

\begin{abstract}
We develop a general criterion for cut elimination in sequent
calculi for propositional modal logics, which rests on absorption of
cut, contraction, weakening and inversion by the purely
modal part of the rule system. Our criterion applies also
to a wide variety of logics outside the realm of normal modal
logic. We give extensive example instantiations of our framework to
various conditional logics. For these, we obtain fully internalised
calculi
which are substantially simpler than those known in the literature,
along with
leaner proofs of cut elimination and complexity.
In one case, conditional logic with modus ponens and conditional
excluded middle,
cut elimination and complexity were explicitly stated as open in the
literature.
\end{abstract}

\maketitle

\section{Introduction}

Cut elimination, originally conceived by Gentzen~\cite{Gentzen34}, is
one of the core concepts of proof theory and plays a major role in
particular for algorithmic aspects of logic, including 
the complexity of automated reasoning and, via
interpolation, modularity issues. The large number of logical calculi
that are currently in use, in particular in various areas of computer
science, motivates efforts to define families of sequent calculi that
cover a variety of logics and admit uniform proofs of cut elimination,
enabled by suitable sufficient conditions. Here, we present such a
method for modal sequent calculi that applies to possibly non-normal
normal modal logics, which appear, e.g.\ in concurrency and knowledge
representation. We use a separation of the modal calculi into a fixed
underlying propositional part and a modal part. The core of our
criterion, that we call the absorption of cut, stipulates that an
application of the cut rule to conclusions of modal rules can be
replaced by a single rule application.
This concept generalises the notion of resolution closed rule
set~\cite{PattinsonSchroder10,SchroderPattinson09},
dropping the assumption that the logic at hand is rank-1, i.e.\
axiomatised by formulas in which the nesting depth of modal operators
is uniformly equal to 1 (such as $K$).

Our method is reasonably simple and intuitive, and nevertheless
applies to a wide range of modal logics. While we use normal modal
logics such as $K$ and $T$ as running examples to illustrate our
concepts at the time of introduction, our main applications are
conditional logics, which have a binary modal operator $\cto$ read as
a non-monotonic implication (unlike default logics, conditional logics
allow nested non-monotonic implications). In particular, we prove
cut-elimination (hence, since the generic systems under consideration
are analytic, the subformula property) for the conditional logic $\CK$
and all its extensions by any of the axioms conditional modus ponens
$(\MP)$ $(A\cto B)\to A\to B$, where $\to$ denotes material
implication, conditional excluded middle $(\CEM)$ $(A\cto
B)\lor(A\cto\neg B)$, and conditional identity $(\ID)$ $A\cto A$ using
our generic procedure. An easy analysis of proof search in the arising
cut-free calculi moreover establishes that the satisfiability problem
of each of these logics is in PSPACE. This is a tight bound for logics
not containing $\CEM$, whereas the provability problem in $\CK\CEM$
and $\CK\ID\CEM$ can be solved in coNP, as we show by a slightly
adapted algorithmic treatment of our calculus using a dynamic
programming approach in the spirit of~\cite{Vardi89}. We point out
that while (different) cut-free labelled sequent calculi for $\CK$,
$\CKMP$, $\CKCEM$, and some further conditional logics, as well as the
ensuing upper complexity bounds, have previously been presented by
Olivetti et al., the corresponding issues for $\CKMPCEM$ have
explicitly been left as open problems~\cite{OlivettiEA07}; moreover,
our coNP upper bounds for $\CK\CEM$ and $\CK\ID\CEM$ improve previous
upper PSPACE bounds.

\paragraph{Related work} A set of sufficient conditions for a sequent
calculus to admit cut elimination and a subsequent analysis of the
complexity \emph{of cut elimination} (not of proof search) is
presented in~\cite{Rasga07}. The range of application of this method
is very wide and encompasses, e.g.\ first-order logic, the modal logic
$S4$, linear logic, and intuitionistic propositional logic. This
generality is reflected in the fact that the method as a whole is
substantially more involved than ours. A simpler method for a
different and comparatively restrictive class of calculi, so-called
canonical calculi, is considered in~\cite{AvronLev01}; this method
does not apply to typical modal systems, as it considers only
so-called \emph{canonical} rules, i.e., left and right introduction
rules for connectives which permit adding a common context
simultaneously in the premise and the conclusion. (In fact, it might
be regarded as the essence of modal logic that its rules fail to be
canonical, e.g.\ the necessitation rule $A/\Box A$ does not generalise
to $\Gamma,A/\Gamma,\Box A$ for a sequent $\Gamma$.) Moreover, the
format of the rules in \emph{op.cit.} does not allow for the
introduction of more than one occurrence of a logical connective,
which is necessary even for the most basic modal logics. The same
applies to \cite{CiabattoniTerui06}.  In \cite{CiabattoniEA08},
logical rules are treated on an individual basis, which precludes the
treatment of cuts between two rule conclusions.  Overall, our notion
of absorption is substantially more general when compared to similar
notions in the papers discussed above, which stipulate that cuts
between left and right rules for the same connective are absorbed by
structural rules.  In our own earlier
work~\cite{PattinsonSchroder10}, we have
considered a special case of the method presented here in the
restricted context of \emph{rank-1} logics; in particular, these
results did not cover logics such as $K4$, $\CKMP$, or $\CKMPCEM$.

This work is an extended and revised version
of~\cite{PattinsonSchroder09}.

\section{Preliminaries and Notation}

A \emph{modal similarity type} (or modal signature) is a set $\Lambda$
of modal operators with associated arities that we keep fixed
throughout the paper. Given a set $V$ of propositional variables, the
set $\Form(\Lambda)$ of $\Lambda$-formulas is given by the grammar
\[ \Form(\Lambda) \ni A, B ::= \bot \mid p \mid \neg A \mid  A \land
B \mid \hearts(A_1, \dots, A_n) \]
where $p \in V$ and $\hearts \in \Lambda$ is $n$-ary. We use
standard abbreviations of the other propositional connectives
$\top$, $\lor$ and $\to$.
%
%
A \emph{$\Lambda$-sequent} is a finite multiset of
$\Lambda$-formulas, and the set of $\Lambda$-sequents is denoted by
$\Seq(\Lambda)$. We write the multiset union of $\Gamma$ and
$\Delta$ as $\Gamma, \Delta$ and  identify a formula $A \in
\Form(\Lambda)$ with the singleton sequent containing only $A$.
If $S \subseteq \Form(\Lambda)$ is a set of formulas, then an
\emph{$S$-substitution} is a  mapping $\sigma: V \to S$. We denote
the result of uniformly substituting $\sigma(p)$ for $p$ in a
formula $A$ by $A \sigma$. This extends pointwise to $\Lambda$-sequents
so that $\Gamma \sigma = A_1 \sigma, \dots, A_n \sigma$ if
$\Gamma = A_1, \dots, A_n$.
If $S \subseteq \Form(\Lambda)$ is a set of $\Lambda$-formulas
and $A \in \Form(\Lambda)$, we say that $A$ is a \emph{propositional
consequence} of $S$ if there exist $A_1, \dots, A_n \in S$
such that $A_1 \land \dots \land A_n \to A$ is a substitution
instance of a propositional tautology.
We write 
$S \entails_\PL A$ if $A$ is a propositional consequence of
$S$ and $A \entails_\PL B$ for $\lbrace A \rbrace \entails_\PL
B$ for the case of single formulas.


%
%
\section{Modal Deduction Systems}
To facilitate the task of comparing the notion of provability in
both Hilbert and Gentzen type proof systems, we introduce the
following notion of a proof rule that can be used, without any
modifications, in both systems. 
\begin{defn}
A
\emph{$\Lambda$-rule} is of the form $\frac{\Gamma_1 \dots
\Gamma_n}{\Gamma_0}$ where $n \geq 0$ and $\Gamma_0, \dots,
\Gamma_n$ are $\Lambda$-sequents. The sequents $\Gamma_1, \dots,
\Gamma_n$ are the
\emph{premises} of the rule and $\Gamma_0$ its \emph{conclusion}. A rule
$\frac{}{\Gamma_0}$
without premises is called a \emph{$\Lambda$-axiom}, which we denote
by just its conclusion, $\Gamma_0$. A \emph{rule set} is just a set
of $\Lambda$-rules, and we say that a rule set $\Rules$ is \emph{substitution
closed}, if $\Gamma_1 \sigma \dots \Gamma_n \sigma / \Gamma_0 \sigma
\in \Rules$ whenever $\Gamma_1 \dots \Gamma_n / \Gamma_0 \in
\Rules$ and $\sigma: V \to \Form(\Lambda)$ is a substitution.
\end{defn}
\noindent
In view of the sequent calculi that we introduce later, we read
sequents disjunctively. Consequently, a rule $\Gamma_1 \dots, \Gamma_n /
\Gamma_0$ can be used to prove the disjunction of $\Gamma_0$, provided
that $\Lor \Gamma_i$ is provable, for all $1 \leq i \leq n$. 
We emphasise that a rule is an expression of the object language,
i.e. it does not contain meta-linguistic variables. As such, it
represents a specific deduction step rather than a family of
possible deductions, which helps to economise on syntactic
categories.
In our examples, concrete rule sets are presented as instances of rule
schemas.

\begin{ex} \label{example:logics-rules} For the modal logics $K$, $K4$
  and $T$, we fix the modal signature $\Lambda = \lbrace \Box \rbrace$
  consisting of a single modal operator $\Box$ with arity one. The
  language of conditional logic is given by the similarity type
  $\Lambda = \lbrace \cto \rbrace$ where the conditional arrow $\cto$
  has arity $2$. We use infix notation and write $A \cto B$ instead of
  $\cto (A, B)$ for $A, B \in \Form(\Lambda)$. Formulas $A\cto B$ are
  interpreted as various forms of conditionals, e.g.\ default
  implication `if $A$ then normally $B$', relevant implication and
  others, depending on the choice of semantics and imposed logical
  principles. Deduction over modal and conditional logics are governed
  by the following rule sets:
\begin{enumerate}[(1)]
\item 
\begin{figure}
\fbox{\begin{minipage}{15cm}

\[ (\N) \frac{A}{\Box A} \qquad (\D){\Box(A \to B) \to (\Box A \to
\Box B)} \qquad (\Four) \Box A \to \Box\Box A \qquad (\R) \Box A \to
A \]

\end{minipage}}
\caption{Axioms and Rules of modal Hilbert Systems}
\label{fig:modal-axioms}
\end{figure}

The rule set $\K$ associated to the modal logic $K$ consists of all
instances of the necessitation rule $(\N)$ and the distribution
axiom $(\D)$ in Figure \ref{fig:modal-axioms}.
The rule sets that axiomatise the logics $T$ and $K4$ arise
by extending this set with
the reflexivity axiom $(\R)$ and the $(\Four)$-axiom, respectively.
We reserve the name $(\T)$ for the reflexivity rule in a cut-free
system.
\item 
\begin{figure}
\setlength{\fboxsep}{2ex}
\fbox{\begin{minipage}{15cm}
\[ (\mathsf{RCEA})\frac{A \otto A'}{(A \cto B) \otto (A' \cto B)} \qquad
  (\mathsf{RCK}) \frac{B_1 \land \dots \land B_n \to B}{(A \cto B_1) \land \dots
    \land (A \cto B_n) \to (A \cto B)}
\]\\[-1ex]
\[
   (\ID) A \cto A \qquad
   (\MP) (A \cto B) \to (A \to B) \qquad
   (\CEM) (A \cto B) \lor (A \cto \neg B)
\]
\end{minipage}}
\caption{Axioms and Rules of conditional Hilbert Systems}
\label{fig:conditional-axioms}
\end{figure}
The basic conditional logic is the system $\CK$ of
  \cite{Chellas80}, axiomatised by the rule set that consists of all
  instances of $(\mathsf{RCEA})$ and $(\mathsf{RCK})$ in Figure
\ref{fig:conditional-axioms}.
  The system $\CK$ constitutes a minimal set of properties to be
  reasonably expected of any conditional, however nonstandard:
  replacement of equivalents in the left hand argument, and
  compatibility with conjunction in the right-hand
  argument. Additional properties are typically imposed when more
  specific interpretations of $\cto$ are intended. E.g.\ the basic
  properties of $\cto$ viewed as a default implication are given by
  Burgess' System $\mathcal{S}$~\cite{Burgess81}, which is related to
  the well-known KLM postulates of default
  reasoning~\cite{KrausLehmannMagidor90}. A treatment of
  System~$\mathcal{S}$ using methods of the present work
  and~\cite{PattinsonSchroder10} is presented in
  \cite{SchroderEA10}. Here, we consider several other standard
  axioms, namely identity $(\ID)$, conditional modus ponens $(\MP)$
and conditional excluded middle $(\CEM)$, also given in Figure
\ref{fig:conditional-axioms}.
  We denote the corresponding extensions of $\CK$ by juxtaposition of
  the respective axioms, e.g. $\CK\MP\CEM$ contains the rules for
  $\CK$ and the axioms $(\MP)$ and $(\CEM)$. As indicated above,
  whether or not these axioms are accepted depends on the intended
  reading of the conditional. E.g., modus ponens is a reasonable
  principle for interpretations of the conditional as a relevant
  implication or as a counterfactual, but not for default implication;
  conditional excluded middle is a controversially discussed property
  of the subjunctive conditional \cite{Cross09}. The
  identity axiom, while accepted for many interpretations of the
  conditional including as default implication, is typically rejected
  for causal interpretations \cite{GiordanoSchwind04}.
\end{enumerate}
Rules with more than one premise arise through saturation of a given
rule set under cut that, e.g. leads to the rules $(\CK_g)$ and
$(\MP_g)$ presented in Section \ref{sec:CL}.
\end{ex}
\noindent
In order to make the mapping between Hilbert-style and Gentzen-style
systems easier, we take the derivability predicate of a
Hilbert-system to be induced by a set of $\Lambda$-rules and read
each sequent as the disjunction of its elements.
The notion of deduction in modal Hilbert systems then takes the
following form.
\begin{defn} \label{defn:hilb-deriv}
Suppose $\Rules$ is a set of rules. The set of $\Rules$-derivable
formulas in the Hilbert-system given by $\Rules$ is the least set of
formulas that
\begin{enumerate}[$\bullet$]
\item contains $A \sigma$ whenever $A$ is a propositional tautology
and $\sigma$ is a substitution
\item contains $B$ whenever it contains $A$ and $A \to B$
\item contains $\Lor \Gamma_0$
whenever it contains $\Lor \Gamma_1, \dots,
\Lor \Gamma_n$ and $\frac{\Gamma_1 \dots \Gamma_n}{\Gamma_0} \in \Rules$.
\end{enumerate}
 We write $\Hilb\Rules \entails A$ if $A$ is
$\Rules$-derivable.
\end{defn}
In other words, the set of derivable formulas is the least set that
contains propositional tautologies, is closed under uniform
substitution, modus ponens and application of rules.
We will later consider Hilbert systems that induce the
same provability predicate based on the following notion of
admissibility.
\begin{defn}\label{defn:admissible}
A rule set $\Rules'$ is \emph{admissible} in $\Hilb\Rules$ if
$\Hilb\Rules \entails A \iff \Hilb(\Rules \cup \Rules')
\entails A$
for all formulas $A \in \Form(\Lambda)$.
Two rule sets $\Rules, \Rules'$ are \emph{equivalent} if
$\Rules$ is admissible in $\Hilb\Rules'$ and $\Rules'$ is admissible
in $\Hilb\Rules$.
\end{defn}
\noindent
In words, $\Rules'$ is admissible in $\Hilb\Rules$ if adding the
rules $\Rules'$ to those of $\Rules$ leaves the set of provable
formulas unchanged.
We note the following trivial, but useful consequence of
admissibility.

\begin{lemma} Let $\Rules$ and $\Rules'$ be equivalent, and let $A \in
  \Form(\Lambda)$. Then $\Hilb\Rules \entails A$ iff $\Hilb\Rules'
  \entails A$.
\end{lemma}
\noindent
The next proposition establishes a rudimentary form of proof
normalisation in 
Hilbert systems and is the key for proving equivalence of
Hilbert and Gentzen-type systems. We show that every derivable
formula in a Hilbert-sytem is a propositional consequence of
conclusions of rules with provable premises, which stratifies proofs
into rule application and propositional reasoning and avoids modus
ponens.
\begin{propn}\label{propn:proof-structure}
Suppose that $S$ is the least set of formulas that is closed under
propositional consequences of rule conclusions,
that is, $S$ contains $A \in
\Form(\Lambda)$ whenever
  there are rules $\Theta_1 / \Gamma_1, \dots, \Theta_n /
  \Gamma_n \in \Rules$ and substitutions $\sigma_1, \dots, \sigma_n:
V
  \to \Form(\Lambda)$ such that $\Lor \Delta \sigma_i \in S$ for all
  $\Delta \in \Theta_i$ ($i = 1, \dots, n$), and
$\lbrace \Lor \Gamma_1 \sigma, \dots, \Lor \Gamma_n \sigma\rbrace
  \entails_\PL A$.
	
Then $S$ coincides with the set of derivable formulas in the
Hilbert-calculus induced by $\Rules$, that is $S = \lbrace A \in
\Form(\Lambda) \mid \Hilb\Rules \entails A
\rbrace$.
\end{propn}
\begin{proof} 
We write $\Hilb\Thm(\Rules) = \lbrace A \in \Form(\Lambda) \mid \Hilb\Rules 
\entails A \rbrace$ for the set of provable formulas in
$\Hilb\Rules$.
The inclusion $S \subseteq \Hilb\Thm(\Rules)$ is immediate as
$\Hilb\Thm(\Rules)$ contains
propositional tautologies, is closed under uniform substitution and
modus ponens. For the reverse inclusion we 
show that
$S$ is closed under $\Rules$-derivability as considered in
Definition \ref{defn:hilb-deriv}.

This is clear for all cases (propositional tautologies, uniform
substitutions, rule application) except possibly modus
ponens. So assume that $\Hilb\Rules \entails A \to B$ and
$\Hilb\Rules \entails A$. By induction hypothesis, there are 
\begin{enumerate}[$\bullet$]
\item Rules $\Theta_1 / \Gamma_1, \dots, \Theta_n / \Gamma_n$ and
substitutions $\sigma_1, \dots, \sigma_n$ such that 
$\lbrace \Lor \Gamma_1 \sigma_1, \dots,  \Lor \Gamma_n \sigma_n
\rbrace
\entails_\PL A \to B$
\item Rules $\Sigma_1 / \Delta_1, \dots, \Sigma_k / \Delta_k$ and
substitutions $\tau_1, \dots, \tau_k$ such that $\lbrace \Lor \Delta_1
\tau_1, \dots,  \Lor \Delta_k \tau_k \rbrace \entails_\PL A$
\end{enumerate}
and moreover $\Lor \Xi \sigma \in S$ whenever $\Xi \in \Theta_1,
\dots, \Theta_n, \Sigma_1, \dots, \Sigma_k$. The claim follows, as
$\lbrace \Gamma_1 \sigma_1,  \dots,  \Gamma_n \sigma_n,  \Delta_1
\tau_1, \dots, \Delta_k \tau_k \rbrace \entails_\PL B$.
\end{proof}

\noindent
In other words, in a modal Hilbert system, each provable formula is
a propositional consequence of rule conclusions with provable
premises. This result forms the basis of our comparison of Hilbert
and Gentzen systems. The point to note is that in a Hilbert system,
provable formulas are propositional consequences of \emph{zero or
more} rule conclusions with provable premises. The propositional reasoning that
is applied when showing that the set of conclusions implies a
formula generally uses the cut-rule. As a consequence, the need for
cut vanishes if there is no need to apply propositional reasoning to
combine conclusions. This is what our notion of cut-absorption
(introduced later in Definition \ref{defn:cut-absorption})
formalises: we show that cut elimination essentially
amounts to the fact that -- in the corresponding Hilbert system --
each valid formula is a consequence of a \emph{at most one} rule
conclusion with provable premise.


We now set the stage for sequent systems that we are going to
address in the remainder of the paper. As we are dealing with
extensions of classical propositional logic, it suffices to work with
a right-handed calculus. Our calculus is equipped with explicit
negation, and therefore precisely dual to modal tableau calculi
\cite{Gore:1999:TMM} that serve as the usual basis for syntactically
determining the complexity of the satisfiability problem.

The notion of
derivability in the sequent calculus associated with a rule set
$\Rules$ is formulated parametric in terms of a set $\Xtra$ of
additional rules that will later be instantiated with relativised
versions of cut, weakening, contraction and inversion.
\begin{defn}
Suppose $\Rules$ and $\Xtra$ are sets of $\Lambda$-rules.
The set of $\Gen\Rules + \Xtra$-derivable
sequents in the Gentzen-system given by $\Rules$ is the least set of
sequents that
\begin{enumerate}[$\bullet$]
\item contains $A, \neg A, \Gamma$ for all sequents $\Gamma \in
\Seq(\Lambda)$ and  formulas $A \in \Form(\Lambda)$
\item contains $\neg \bot, \Gamma$ for all $\Gamma \in \Seq(\Lambda)$
\item is closed under instances of the rule schemas
\[ (\neg\land)\frac{\Gamma, \neg A, \neg B}{\Gamma, \neg(A \land B)} \qquad
   (\land)\frac{\Gamma, A \quad \Gamma, B}{\Gamma, A \land B} \qquad
	 (\neg) \frac{\Gamma, A}{\Gamma, \neg \neg A} \qquad
\] 
where $A \in \Form(\Lambda)$ ranges over formulas and $\Gamma
\subseteq \Form(\Lambda)$ over multisets of formulas.  We call
the above rules the \emph{propositional rules} and the formula
occurring in the conclusion but not in $\Gamma$ \emph{principal}
in the respective rule.
\item is closed under the rules in $\Rules \cup \Xtra$, i.e. it
contains $\Gamma_0$ whenever it contains $\Gamma_1, \dots, \Gamma_n$
and $\frac{\Gamma_1
\dots \Gamma_n}{\Gamma_0} \in \Rules \cup \Xtra$.
\end{enumerate}
We write $\Gen\Rules  + \Xtra \entails \Gamma$ if $\Gamma$ can be derived
in this way and $\Gen\Rules \entails\Gamma$ if $\Xtra =
\emptyset$. As for Hilbert-style calculi (Definition
\ref{defn:admissible}), we call a rule set $\Rules'$
\emph{admissible} in $\Gen\Rules$ in case $\Gen\Rules \entails \Gamma
\iff \Gen(\Rules \cup \Rules') \entails \Gamma$ for all $\Gamma \in
\Seq(\Lambda)$.
\end{defn}
\noindent
The set $\Xtra$ of extra rules will later be instantiated with
a relativised version of the cut rule and additional axioms that
locally capture the effect of weakening, contraction and inversion, applied
to rule premises. This allows to formulate \emph{local} conditions
for the admissibility of cut that can be checked on a per-rule
basis.

Many other formulations of sequent systems only permit axioms of
the form $\Gamma, p, \neg p$ where $p \in V$ is a propositional
atom. The reason for being more liberal here is that this makes it
easier to prove admissibility of uniform substitution, at the
expense of losing depth-preserving admissibility of structural
rules. We come back to this matter in Remark \ref{rem:congruence}.
The following proposition is readily established by an induction on
the provability predicate $\Hilb\Rules \entails$.
\begin{propn} \label{propn:gen-hilb}
Suppose $\Gamma \in \Seq(\Lambda)$ is a sequent. Then $\Hilb\Rules \entails
\Lor \Gamma$ if $\Gen\Rules \entails \Gamma$.
\end{propn}
\noindent
The remainder of the paper is concerned with the converse of the
above proposition, which relies on specific properties of the rule
set $\Rules$.
\section{Generic Modal Cut Elimination}
In order to establish the converse of Proposition \ref{propn:gen-hilb} we
need to establish that the cut rule is admissible in the Gentzen 
system $\Gen\Rules$ defined by the ruleset $\Rules$.  Clearly, we
cannot expect that cut elimination holds in general: it is well
known (and easy to check) that the sequent system arising from the
rule set consisting of all instances of $(\N)$ and $(\D)$, presented
in Example \ref{example:logics-rules} does \emph{not}
enjoy cut elimination. In other words, we have to look for
constructions that  allow us to transform a given rule set into one
for which cut elimination holds.
The main result of our analysis is that cut elimination holds
if the rule set under consideration satisfies four crucial
requirements that are local in the sense that they can be checked on
a per-rule basis without the need of carrying out a fully-fledged
cut-elimination proof: absorption of weakening, contraction,
inversion and cut.

The first three properties can be checked  for each rule
individually and amount to the admissibility of the respective
principle, and the last requirement amounts to the
possibility of eliminating cut between a pair of rule conclusions.
We emphasise that these properties can be checked locally for the
modal rules, and cut elimination will follow automatically.
It is not particularly  surprising that cut elimination holds under these
assumptions. However, isolating the four conditions above provides
us with means to convert a modal Hilbert system into an equivalent
cut-free sequent calculus.
\noindent
We now introduce relativised versions of the structural rules that
will be the main tool in the proof of cut elimination. This can be
seen as permutability of structural rules: every derivation of
$\Gamma$ 
from premises $\Gamma_1, \dots, \Gamma_n$ that ends in weakening,
inversion or contraction is applied can be replaced by a derivation
of $\Gamma$  where weakening, inversion and
contraction is only applied to the premises $\Gamma_1, \dots,
\Gamma_n$.
\begin{defn} \label{defn:structural}
Suppose $\Gamma$ is a $\Lambda$-sequent and let $\Struct(\Gamma)$ consist of
the axioms
\begin{enumerate}[$\bullet$]
\item $\Gamma, A$ for all $A \in \Form(\Lambda)$
\item $\Delta, A$ if $\Gamma = \Delta, A, A$ for some $\Delta \in
\Seq(\Lambda), A \in
\Form(\Lambda)$
\item $\Delta, A$ if $\Gamma = \Delta, \neg \neg A$ for some
$\Delta \in \Seq(\Lambda), A \in
\Form(\Lambda)$
\item $\Delta, \neg A_1, \neg A_2$ if $\Gamma = \Delta, \neg(A_1
\land A_2)$ for some $\Delta
\in \Seq(\Lambda)$, $A_1, A_2 \in \Form(\Lambda)$
\item $\Delta, A_i$ for $i = 1, 2$ if $\Gamma = \Delta, (A_1 \land
A_2)$ for some $\Delta
\in \Seq(\Lambda)$, $A_1, A_2 \in \Form(\Lambda)$
\end{enumerate}
We say that a rule set $\Rules$ \emph{absorbs the structural rules}
if
\[ \Gen\Rules + \Struct(\Gamma_1) \cup \dots \cup \Struct(\Gamma_n)
\entails \Gamma 
\]
for all $\frac{\Gamma_1 \dots \Gamma_n}{\Gamma_0} \in \Rules$ and
all $\Gamma \in
\Struct(\Gamma_0)$.
\end{defn}
\noindent
In other words, a deduction step that applies weakening, contraction
or inversion to a rule conclusion can be replaced by a (possibly
different)
rule where the corresponding structural rules are applied to the
premises. We discuss a number of standard examples before stating
that absorption of the structural rules implies their admissibility.

\begin{ex}  \label{example:absorbing-rule-sets}
The rule sets containing all instances of either of the following
rule schemas $(\K)$,  $(\T)$, and $(\Kfour)$, respectively,
\begin{equation*}
(\K)\;\frac{\neg A_1, \dots, \neg A_n, A_0}{\neg \Box A_1, \dots,
\neg \Box A_n, \Box A_0, \Gamma}\quad
(\T) \;
 \frac{\neg A, \neg \Box A, \Gamma}{\neg \Box A, \Gamma}\quad
(\Kfour)\;
\frac{\neg A_1, \neg \Box A_1, \dots, \neg A_n, \neg \Box A_n,
B}
  {\neg \Box A_1, \dots, \neg \Box A_n, \Box B, \Gamma}
\end{equation*}
absorb the structural rules. We note that $(\K)$ absorbs weakening
due to the presence of $\Gamma$ in the conclusion, and the
absorption of contraction in $(\T)$ and $(\Kfour)$ is a consequence of
the presence of the negated $\Box$-formulas in the premise.
The absorption of inversion in a consequence of the weakening
context $\Gamma$ in $(\K)$ and $(\Kfour)$ and implied by duplicating the
context $\Gamma$ in $(\T)$.
On the other hand, the rule sets defined by
\[ \frac{\neg A_1, \dots, \neg A_n, A_0}{\neg \Box A_1, \dots,
\neg \Box A_n, \Box A_0} \qquad
\qquad
\frac{\neg A, \Gamma}{\neg \Box A, \Gamma}
\] fail to absorb the structural rules: the rule on the left fails
to absorb weakening, whereas the right-hand rule does not absorb
contraction.
\end{ex}

\noindent
It should be intuitively clear that absorption of structural rules
implies their admissibility, which we establish next. 
\begin{propn} \label{propn:struct-admissible}
Suppose $\Rules$ absorbs the structural rules. Then all instances of
the rule schemas of weakening, contraction and inversion
\[ \frac{\Gamma}{\Gamma, A} \qquad \frac{\Gamma, A, A}{\Gamma, A}
   \qquad
	 \frac{\Gamma, \neg \neg A}{\Gamma, A}
	 \qquad
	 \frac{\Gamma, \neg(A_1 \land A_2)}{\Gamma, \neg A_1, \neg A_2}
	 \qquad
	 \frac{\Gamma, A_1 \land A_2}{\Gamma, A_i} (i = 1, 2)
\] 
where $\Gamma \in \Seq(\Lambda)$ and $A, A_1, A_2 \in \Form(\Lambda)$ are
admissible in $\Gen\Rules$.
\end{propn}
\begin{proof}
Standard induction on proofs in $\Gen\Rules$
where the case of propositional rules is standard and the inductive
case for modal rules immediately follows from absorption.
\end{proof}

\begin{rem} \label{rem:congruence}\hfill
\begin{enumerate}[(1)]
\item
The main purpose for introducing the notion of absorption of
structural rules (Definition \ref{defn:structural}) is to have a
handy criterion that guarantees admissibility of the structural
rules (Proposition \ref{propn:struct-admissible}). Our definition
offers a compromise between generality and simplicity. In essence,
a rule set absorbs structural rules if an application of weakening,
contraction or inversion can be pushed up one level of the proof
tree. A weaker version of Definition \ref{defn:structural} would
require that an application of weakening, contraction or inversion
to a rule conclusion can be replaced by a sequence of deduction
steps where the structural rule in question can not only be applied
to the premises of the rule, but also freely anywhere else, provided
that these additional applications are smaller in a well-founded
ordering.  However, we are presently not aware of any examples where
this extra generality would be necessary.
\item  In many sequent systems, the statement of Proposition
\ref{propn:struct-admissible} can be strengthened to say that
weakening, contraction and inversion are depth-preserving
admissible, i.e. does not increase the height of the proof tree.
This is in general false for the systems considered here as axioms
are of the form $A, \neg A, \Gamma$ for $A \in \Form(\Lambda)$ and,
for instance, $(A \land B), \neg (A \land B)$ is derivable with a
proof of height one (being an axiom), but, e.g. $A \land B, \neg A,
\neg B$ cannot be established by a proof of depth one (not being an
	axiom). It is easy to see that weakening, inversion and
	contraction
	are in fact depth-preserving admissible if only atomic axioms of the
	form $p, \neg p, \Gamma$ are allowed, for $p \in V$ a propositional
	variable. The more general form of axioms adopted in this paper
	allows us to simplify many constructions as we do not have to
	consider a congruence rule explicitly which would allow us to prove
	(rather than to assume as axioms) sequents of the form $\Box A, \neg \Box A,
	\Gamma$.
	\end{enumerate}
	\end{rem}
	\noindent
	Having dealt with the structural rules, we now address our main
	concern: the admissibility of the cut rule. In contrast to the
	absorption of structural rules, we need one additional degree of
	freedom in that we need to allow ourselves to apply cut to a
	structurally smaller formula.

	\begin{defn}\label{defn:cut-absorption}
	The \emph{size} of a formula $A \in \Form(\Lambda)$ is given
	inductively by $\size(p) = \size(\bot) = 1$, $\size(A \land B) =
	1 + \size(A) + \size(B)$, $\size(\neg A) = 1 + \size(A)$ and, for the modal
	case, $\size(\hearts(A_1,
	\dots, A_n)) = 1 + \size(A_1) + \dots + \size(A_n)$.

A ruleset $\Rules$ \emph{absorbs cut}, if for all rules
$(r_1)\frac{\Gamma_1 \dots \Gamma_n}{A, \Gamma_0}$, $(r_2) \frac{\Delta_1
\dots \Delta_k}{\neg A,
\Delta_0} \in \Rules$
\[ \Gen\Rules + \Cut(A, r_1, r_2) \entails
\Gamma_0, \Delta_0 \]
where $\Cut(A, r_1, r_2)$ consists of all instances of the rule schemas
\[ \frac{\Gamma, C \quad \Delta, \neg C}{\Gamma, \Delta} 
\quad
 \frac{\Gamma}{\Gamma, A} \qquad \frac{\Gamma, A, A}{\Gamma, A}
   \quad
	 \frac{\Gamma, \neg \neg A}{\Gamma, A}
	 \quad
	 \frac{\Gamma, \neg(A_1 \land A_2)}{\Gamma, \neg A_1, \neg A_2}
	 \quad
	 \frac{\Gamma, A_1 \land A_2}{\Gamma, A_i} 
\]
where $\size(C) < \size(A)$ in the leftmost rule and $i = 1, 2$ in
the rightmost schema, together with
the axioms $\Gamma_1, \dots, \Gamma_n, \Delta_1, \dots, \Delta_k$
and all sequents of the form $\Gamma, \Delta$ where $\Gamma, \Delta \in
\Seq(\Lambda)$ and,
for some $B \in \Form(\Lambda)$,
\begin{enumerate}[$\bullet$]
\item $\Gamma, B$ and $\Delta, \neg B
\in \lbrace \Gamma_1, \dots, \Gamma_n, \Delta_1, \dots, \Delta_k
\rbrace$, or
\item $\Gamma, B = \Gamma_0, A$ and $\Delta, \neg B \in \lbrace
\Delta_1, \dots, \Delta_k \rbrace$, or
\item $\Gamma, B = \Delta_0, \neg A$ and $\Delta, \neg B \in \lbrace
\Gamma_1, \dots, \Gamma_n \rbrace$.
\end{enumerate}
A rule set
that absorbs structural rules and the cut rule is called
\emph{absorbing}.
\end{defn}
\noindent
The intuition behind the above definition is similar to that of
absorption of structural rules,
but we have two additional degrees of freedom: we can not
only apply the cut rule to rule premises, but we can moreover freely
use both cut on structurally smaller formulas and the structural rules.
This allows us to use
the standard argument, a double induction on the structure of
the cut formula and the size of the proof tree, to
establish cut elimination. This is carried out in the proof of the
next theorem.

\begin{thm} \label{thm:cut-elim}
Suppose $\Rules$ is absorbing. Then the cut rule 
\[ \frac{\Gamma, A \qquad \Delta, \neg A}{\Gamma, \Delta} \]
is admissible in $\Gen\Rules$.
\end{thm}
\begin{proof}
We use Gentzen's classical method and proceed by a double induction
on the size of the cut formula and the size of the proof tree.
That is, we prove the statement
\[ \forall A \in \Form(\Lambda) \forall n \in \omega (n = n_1 + n_2
\mbox{ \& } \entails_{n_1} \Gamma, A \mbox{ \& } \entails_{n_2} \Delta, \neg A \implies
 \entails \Gamma, \Delta) \]
by induction on $\size(A)$
where, in the inductive step, we use a side
induction on the size of proof trees, as indicated by the subscript
of the entailment sign. Formally, the
relation $\entails_n$ is defined inductively by $\entails_1 \Gamma, A, \neg A$
and 
\[ \frac{\entails_n A}{\entails_{n+1} \neg \neg A} 
   \qquad 
	 \frac{\entails_n \Gamma, \neg A, \neg B}{\entails_{n+1} \Gamma, \neg(A \land B)} 
   \qquad
	 \frac{\entails_n \Gamma, A \quad \entails_k \Gamma, B}{\entails_{n + k + 1} \Gamma, A \land B} 
	 \qquad 
	 \frac{\entails_{n_1} F_1 \dots \entails_{n_k} F_k}{\entails_{n_1 + \dots + n_k + 1} F_0}
\]
where, in the last rule, $\frac{F_1 \dots F_k}{F_0} \in \Rules$.  
We may inductively assume that the statement holds for all cut
formulas $C < A$ and to prove the statement for $A$ we have to
consider the following cases:
\begin{enumerate}[(1)]
\item \label{cut:rule-rule} 
cuts that arise between two rule conclusions
\item \label{cut:rule-prop} cuts that arise between a rule conclusion and the conclusion
of a propositional rule or axiom
\item \label{cut:prop-prop} cuts that arise between two propositional rules.
\end{enumerate}
We start with item (\ref{cut:rule-rule}), which follows directly
from the fact that $\Rules$ absorbs cut. In more detail, suppose
that $\frac{F_1 \dots F_k}{F_0, A}$ and $\frac{G_1 \dots G_k}{G_0,
\neg A}
\in \Rules$ and $\entails_{n_i} F_i$ ($i = 1, \dots, k$) and
$\entails_{m_j} G_j$ for $j = 1, \dots, l$.  As $\Rules$ absorbs
cut, we have that $F_0, G_0$ is derivable using cuts on formulas $<
A$ from the additional assumptions $\Gamma, \Delta$ provided that
for some $D \in \Form(\Lambda)$ we have that both $\Gamma, D$ and
$\Delta, \neg D$ are among the $F_1, \dots, F_k, G_1, \dots, G_l$.
In case $\Gamma, D = \Delta, \neg D$ we have that $\Gamma \subseteq
\Gamma, \Delta$ and $\entails \Gamma, \Delta$ as weakening is
admissible in $\Gen\Rules$. 
Assuming that $\entails_x \Gamma, D$ and $\entails_y \Delta, \neg D$
for $\Gamma, D \neq \Delta, \neg D$
we have that $x + y < 2 + \sum_i n_i + \sum_j m_j$ and hence
$\entails \Gamma, \Delta$ by (inner) induction hypothesis. The fact
that -- in the deduction of $F_0, G_0$ -- we may also have to use
cuts on formulas $< A$ is discharged by the outer induction
hypothesis and possible uses of weakening, contraction and inversion
are admissible by Proposition \ref{propn:struct-admissible}.

As regards item (\ref{cut:rule-prop}) we only discuss a subset of
the cases that showcase the need for contraction, weakening and
inversion to be admissible. For the whole discussion, suppose that
$\frac{F_1 \dots F_k}{F_0} \in \Rules$ and $\entails_{n_i} F_i$ for
$i = 1, \dots, k$.
\begin{enumerate}[$\bullet$]
\item Suppose that $F_0 = F_0', A$ and $G_0, \neg A$ is an axiom. In
case $A \in G_0$ we have that $F_0 = F_0', A \subseteq F_0', G_0$
and $\entails F_0', G_0$ follows from $\entails F_0', A$ as
$\Gen\Rules$ admits weakening. In case $\neg A \notin G_0$ we have
that $G_0$ is an axiom, and hence so is $G_0, F_0'$.
\item Suppose that $F_0 = F_0', A$ and $\neg A, G_0$ has been
derived using $(\neg\land)$. We have to discuss two cases, depending
on whether or not $\neg A$ is principal in the application of
$(\neg\land)$.

\noindent\emph{Case $A = A' \land B'$ and $\entails_{m} \neg
A', \neg B', G_0$ so that $\entails_{m+1} \neg A, G_0$.} As $\Rules$ absorbs structural rules, we have
that $\Gen\Rules \entails F_0', A$ and $\Gen\Rules \entails F_0',
B$. As cuts on $A'$ and $B'$ can be eliminated by induction
hypothesis, we have $\Gen\Rules \entails F_0', F_0', G_0$ and
therefore $\Gen\Rules \entails F_0', G_0$ as $\Gen\Rules$ admits
contraction.

\noindent\emph{Case $\entails_m \neg C, \neg D, \neg A, G_0$ so that
$\entails_{m+1} \neg(C \land D), \neg A, G_0$.} As $m + 1 + \sum_{i
= 1}^k n_i < m + 1 + 1 + \sum_{i=1}^k n_i$ we may apply the inner
induction hypothesis to conclude $\entails \neg C, \neg D, G_0,
F_0'$ and applying $(\neg\land)$ gives $\entails F_0', \neg(C \land
D), G_0$.
\end{enumerate}
\noindent
All the other cases follow exactly the same pattern. We now focus on
item (\ref{cut:prop-prop}), that is, we show how cuts between the
conclusions of propositional rules and axioms can be eliminated.
This is mostly standard and again we only discuss a subset of the
cases. Suppose that $\entails_n F_0, A$ and $\entails_m G_0, \neg
A$.
\begin{enumerate}[$\bullet$]
\item If both $F_0, A$ and $G_0, \neg A$ are axioms, then so is
$F_0,G_0$.
\item Suppose that $F_0, A$ has been derived using $(\land)$ and
$G_0, \neg A$ has been derived using $(\neg \land)$. We distinguish
four cases depending on whether or not $A$ is principal in the
application of $(\land)$ or $(\neg\land)$.

\noindent\emph{Case $A = A' \land B'$ and $\entails_{n_0} F_0, A$,
$\entails_{n_1} F_0, B'$ so that $n = n_0 + n_1 + 1$ and $\entails_n
A \land B, F_0$.} 
If $A$ is principal in the application of $(\neg\land)$, we have
that $\entails_{m-1} G_0, \neg A', \neg B'$. By (outer) induction
hypothesis, cuts on $A'$ and $B'$ can be eliminated so that we have
$\entails F_0, F_0, G_0$ and it follows from closure under
contraction that $\entails F_0, G_0$.

If $A$ is not principal in the application of $(\neg \land)$ we have
that $\entails_{m-1} \neg C, \neg D, \neg(A' \land B'), G_0'$ so
that $G_0 = \neg(C \land D), G_0'$ and $\entails_m \neg(A \land B),
G_0'$. As $\entails_n F_0, A$ and $\entails_{m-1} \neg C, \neg D,
\neg A, G_0'$ and $n + (m-1) < n + m$ we can apply the inner
induction hypothesis to eliminate the cut on $A$ so that 
$\entails F_0, \neg C, \neg D, G_0'$ and applying $(\neg\land)$
gives $\entails F_0, \neg(C \land D), G_0' = F_0, G_0$ as required.
The two cases where $A$ is not principal in the application of
$(\neg\land)$ follow exactly the same pattern.
\end{enumerate}
The remaining cases of cuts between propositional rules and axioms
are entirely analogous, and therefore omitted.
\end{proof}
\noindent
We illustrate the preceding theorem by using it to derive the
well-known fact that cut-elimination holds for the 
modal logics $\K, \Kfour$ and $\T$ and use it to derive
cut-elimination for various conditional logics in  Section
\ref{sec:CL}.
\begin{ex} \label{example:K}
The rule sets $\K, \Kfour$ and $\T$ are absorbing. We have already
seen that they absorb weakening, contraction and inversion in Example
\ref{example:absorbing-rule-sets} so everything that remains to be seen is
that they also absorb cut. For $(K)$, we need to apply cut to a
formula of smaller size. For the 
two instances
\[ 
(r_1) \frac{\neg A_1, \dots, \neg A_n, A_0}{\neg \Box A_1, \dots, \neg \Box A_n, \Box A_0, \Gamma} 
\qquad
(r_2)\frac{\neg B_1, \dots, \neg B_k, B_0}{\neg \Box B_1, \dots, \neg \Box B_k, \Box B_0, \Delta} 
\]
we need to consider, up to symmetry,  the cases $A_i = B_0$, $\Box A_i
\in \Delta$ and $\neg \Box A_0 \in \Delta$, for $i = 1, \dots, n$.
Here, we only treat the first case for $i = 1$ where we have to show that
$\neg \Box A_2, \dots, \neg
\Box A_n, \Box A_0,
\neg \Box B_1, \dots, \neg \Box B_k, \Gamma, \Delta$ is derivable
from 
$\Gen\Rules + \Cut(\Box A_1, r_1, r_2)$,
which follows as the latter
system allows us to apply cut on $A_1 = B_0$. The case $\Box A_i \in
\Delta$ and $\neg \Box A_0 \in \Delta$ are straight forward.

The argument to show that $(\K\Four)$ is absorbing is similar, and uses
an additional (admissible) instance of cut on a formula of smaller
size and contraction. For $(\T)$ we only consider instances of cut between two
conclusions  of
\[ (r_1)\frac{\neg A, \neg \Box A, \Gamma}{\neg \Box A, \Gamma} \qquad
  (r_2) \frac{\neg B, \neg \Box B, \Delta}{\neg \Box B, \Delta}
\]
of the $\T$-rule.  We only demonstrate the case 
$\Box A  \in \Delta$. In this case, $\Delta = \Delta',
\Box A$ and we have to show that $\neg \Box B, \Gamma, \Delta'$ can
be derived in $\Cut(\Box A, r_1, r_2)$.
The latter system allows us to cut $\neg \Box A$
between the conclusion of $(\T)$ on the left and the premise of the
right hand rule, i.e., we have that
$\Cut(\Box A, r_1, r_2)
\entails \neg B, \neg \Box B, \Gamma, \Delta')$ and an
application of $(\T)$ now gives derivability of $\neg \Box B, \Gamma,
\Delta'$.
\end{ex}

\section{Equivalence of Hilbert and Gentzen Systems}

We now investigate the relationship between provability in a
Hilbert-system and provability in the associated Gentzen system.
We note the following standard lemmas that we will use later on.
\begin{lemma} \label{lemma:taut-subst}
Suppose that $\Lambda$ is a modal similarity type and $\Rules$ is a
set of $\Lambda$-rules.
\begin{enumerate}[\em(1)]
\item Let $A \in \Form(\Lambda)$ be a propositional tautology. Then
$\Gen\Rules\entails A$.
\item Let $\Rules$ be closed under substitution and $\Gamma \in
\Seq(\Lambda)$. Then
$\Gen\Rules \entails \Gamma \sigma$ whenever
$\Gen\Rules \entails \Gamma$.
\end{enumerate}
\end{lemma}


%

\begin{rem}
Being able to prove the previous lemma is the main reason for
formulating axioms as $A, \neg A, \Gamma$ where $A \in
\Form(\Lambda)$ rather than $p, \neg p, \Gamma$. Both formulations
are equivalent if the modal congruence rule 
\[ \frac{A_1 \otto A_1' \qquad \dots \qquad A_n \otto A_n'}{\hearts(A_1, \dots,
A_n)
\to \hearts(A_1', \dots, A_n')} \]
is admissible. However, Lemma \ref{lemma:taut-subst} can be proved
without the assumption that congruence is admissible using axioms of
the form $A, \neg A, \Gamma$.
\end{rem}

\begin{thm} \label{thm:gen-hilb-equiv}
Suppose $\Rules$ is absorbing and substitution closed. Then 
$\Gen\Rules \entails \Gamma \iff \Hilb\Rules \entails \Lor
\Gamma$
for all $\Gamma \in \Seq(\Lambda)$.
\end{thm}

\begin{proof}[Sketch]
We only need to show the direction from right to left. Inductively
assume that $\Hilb\Rules \entails \Lor \Gamma$ for $\Gamma \in
\Seq(\Lambda)$. By Proposition
\ref{propn:gen-hilb} we have that there are rules
$\Theta_i/\Gamma_i$ and substitutions $\sigma_i$, $i = 1, \dots, n$ such
that 
\begin{enumerate}[$\bullet$]
\item $\Hilb\Rules \entails \Delta \sigma_i$ whenever $\Delta \in
\Theta_i$ ($i = 1, \dots, n$)
\item $\lbrace \Lor \Gamma_1 \sigma_1,  \dots, \Lor \Gamma_n
\sigma_n \rbrace
\entails_\PL \Lor \Gamma$. 
\end{enumerate}
By induction hypothesis, $\Gen\Rules \entails \Delta \sigma_i$ for
all $i = 1, \dots, n$ and $\Delta \in \Theta_i$. By Lemma
\ref{lemma:taut-subst} we have
\[ \Gen\Rules \entails \Lor \Gamma_1 \sigma_1 \land \dots \land \Lor
\Gamma_n \sigma_n \to \Lor \Gamma. \]
The claim follows by applying cut, contraction and inversion.
\end{proof}
%
%
\noindent
The construction of an absorbing rule set from a given set of axioms
and rules essentially boils down to adding the missing instances of
cut, weakening, contraction and inversion to a given rule set.  The
soundness of this process is witnessed by the following two trivial
lemmas (both of which rest on the fact that $\Hilb\Rules$ incorporates
full propositional reasoning). We use these to derive an absorbing
rule set for $K$ in the present section, and to establish
cut-elimination for a large range of conditional logics in the next
section.
\begin{lemma} \label{lemma:adm-cut}
{\sloppy
Suppose $\Gamma_1, \dots, \Gamma_n / \neg A, \Gamma_0$ and
$\Delta_1, \dots, \Delta_k / A, \Delta_0 \in \Rules$. Then the rule \break
$\Gamma_1, \dots, \Gamma_n, \Delta_1, \dots, \Delta_k / \Gamma_0,
\Delta_0$ is admissible in $\Hilb\Rules$.}
\end{lemma}
%
%
\noindent
The same applies to instances of the structural rules of weakening,
contraction and inversion. As we wish to extend the rule set while
leaving the provability predicate in the Hilbert calculus unchanged,
the following formulation is handy for our purposes -- in particular
it implies that we can freely use structural rules both in the premise
and conclusion.
\begin{lemma} \label{lemma:adm-prop}
Let $\Gamma_1, \dots, \Gamma_n / \Gamma_0 \in \Rules$. If
$\Delta_0, \dots, \Delta_k \in \Seq(\Lambda)$ and both
\[ \lbrace \Lor \Delta_1, \dots, \Lor \Delta_k \rbrace \entails_\PL
\Lor \Gamma_i (1 \leq i \leq n) \mbox{ and } \Lor \Gamma_0 \entails_{\PL}
\Lor \Delta_0 \]
then the rule $\Delta_1, \dots, \Delta_k / \Delta_0$ is admissible
in $\Hilb\Rules$. 
\end{lemma}
This gives us a recipe for constructing  rule sets that absorb
contraction and cut: simply add more rules according to the 
lemmas above. This will not change the notion of provability in the
Hilbert system, but when this process terminates, the ensuing rule
set will be absorbing and  gives rise to a cut free sequent calculus.
\begin{ex}[Modal Logic $K$]
In a Hilbert-style calculus, the axiomatisation of $K$ is usually
described in terms of the distribution axiom (which we view as a
rule with empty premise) and the necessitation rule:
\[ (\D) \quad \Box (A \to B) \to \Box A \to \Box B  \qquad (\N)
\frac{A}{\Box A} \]
We first apply Lemma \ref{lemma:adm-cut} to break the propositional
connectives in the distribution axiom. We have that the axiom
$\neg \Box(A \to B), \neg \Box A, \Box B$ is admissible by Lemma
\ref{lemma:adm-prop}, and applying Lemma \ref{lemma:adm-cut} to this
axiom and the instance $A \to B / \Box(A \to B)$ of the
necessitation rule gives admissibility of
the all instances of
\[ \frac{\neg A, B}{\neg \Box A, \Box B} \]
with the help of (admissible) propositional reasoning in the
premise. The same procedure, applied to the instances
\[ \frac{\neg A, B \to C}{\neg \Box A, \Box(B \to C)} \qquad
\neg \Box(B \to C), \neg \Box B, \Box C \]
gives admissibility of the left hand rule below,
\[\frac{\neg A, \neg B, C}{\neg \Box A, \neg \Box B, \Box C} \qquad
   \frac{\neg A_1, \dots, \neg A_n, A_0}{\neg \Box A_1, \dots, \neg
\Box A_n, \Box A_0, \Gamma} \]
and continuing in this way and absorbing weakening, we obtain admissibility of
the rule on the right,
where $\Gamma \in \Seq(\Lambda)$ is an arbitrary context. We have shown
previously that this rule set is absorbing, and it is easy to see
that it is equivalent to the rule set consisting of all instances of
$(\N)$ and $(\D)$.
\end{ex}

\section{Applications: Sequent Calculi for Conditional Logics}
\label{sec:CL}

After having seen how the construction of absorbing rule sets gives
rise to cut-elimination for a number of
well-studied normal modal logics,  in this section we construct a
cut-free sequent calculus for a number of conditional logics. 

Conditional logics \cite{Chellas80} are extensions of propositional
logic by a non-monotonic conditional $A \cto B$, read as ``$B$ holds
under the condition that $A$''. Formulas of the form $A\cto B$ or
$\neg(A\cto B)$ are called \emph{conditional literals}, and in such a
conditional literal, we refer to $A$ as the \emph{(conditional)
  antecedent} and to $B$ as the \emph{(conditional) consequent}. The
conditional implication is non-monotonic in general, i.e.\ the
validity of $A \cto B$ does \emph{not} imply that $(A \land C) \cto B$
is also a valid statement.

Axiomatically, the first argument $A$ of the conditional operation $A
\cto B$ behaves like the $\Box$ in neighbourhood frames and only
supports replacement of equivalents, whereas the second argument $B$
obeys the rules of $K$.  We recall from Example
\ref{example:logics-rules} (see also Figure
\ref{fig:conditional-axioms}) that  
$\CK$ is axiomatised by the rules $(\mathsf{RCEA})$ and
$(\mathsf{RCK})$ that we augment with a subset of $(\ID)$, $(\MP)$
and $(\CEM)$. For each system, we apply Lemma \ref{lemma:adm-prop} and Lemma
\ref{lemma:adm-cut} to the given rule sets repeatedly to generate
new rules that are automatically sound over the original Hilbert
system. This procedure leads to the rules summarised in Figure
\ref{fig:conditional-rules} where we have used the following
notational shorthand to express the equivalences in the premise of
$\CK$:
\begin{notn} \label{notn:eq}
If $A_0, \dots, A_n \in \Form(\Lambda)$ are conditional formulas, we
write $A_0 = \dots = A_n$ for the sequence of sequents consisting of
$\neg A_0, A_i$ and $\neg A_i, A_0$ for all $1 \leq i \leq n$.
\end{notn}
\noindent
We now discuss the arising system in detail, and start with those
not containing $(\CEM)$
and then proceed to add $(\CEM)$ as an additional principle. 
\noindent
\begin{figure}
\setlength{\fboxsep}{2ex}
\setlength{\parskip}{\medskipamount}
\fbox{\begin{minipage}{15cm}
\[ (\CK_g) \frac{A_0 = \dots = A_n \quad \neg B_1, \dots, \neg B_n,
B_0}{\neg (A_1 \cto B_1), \dots, \neg (A_n \cto B_n), (A_0 \cto
B_0), \Gamma}
\] 

\vspace*{1ex}
\[ (\CKID_g) \frac{A_0 = \dots = A_n \quad \neg A_0, \neg B_1,
\dots, \neg B_n, B_0}
   {\neg (A_1 \cto B_1), \dots, \neg(A_n \cto B_n), (A_0 \cto B_0),
	 \Gamma}
\] 

\vspace*{1ex}
\[
  (\MP_g) 
	\frac{A, \neg(A \cto B), \Gamma \qquad \neg B, \neg (A \cto B),
	\Gamma}{\neg (A \cto B), \Gamma}
\] 

\vspace*{1ex}
\[ (\CK\CEM_g) \frac{A_0 = \dots = A_n \quad B_0, \dots, B_j, \neg B_{j+1}, \neg
B_n}{(A_0 \cto B_0), \dots, (A_j \cto B_j), \neg (A_{j+1} \cto
B_{j+1}), \dots, \neg (A_n \cto B_n), \Gamma} 
\quad(0 \leq j \leq n) \]

\vspace*{1ex}
\[ (\CK\CEM\ID_g) 
\frac{A_0 = \dots = A_n \quad \neg A_0, B_0, \dots, B_j, \neg B_{j+1}, \neg
B_n}{(A_0 \cto B_0), \dots, (A_j \cto B_j), \neg (A_{j+1} \cto
B_{j+1}), \dots, \neg (A_n \cto B_n), \Gamma} 
\]

\vspace*{1ex}
\[ (\MPEM_g) \frac{A, (A \cto B), \Gamma \qquad B, (A \cto B),
\Gamma}{(A \cto B), \Gamma} \] 
\end{minipage}}
\caption{Cut-Free Conditional Sequent Rules}
\label{fig:conditional-rules}
\end{figure}
For each system, we show cut-free completeness
and develop the format of the respective rules as we go along. In
summary, we obtain the following cut-free sequent calculi for
extensions of $(\CK)$ summarised in Figure
\ref{fig:conditional-calculi}.
\begin{figure}
%
\begin{tabular}{|l|l|} \hline
Hilbert System & Sequent System  \\ \hline \hline 
$\CK$  & $\CK_g$ \\ \hline
$\CKID$ &  $\CKID_g$ \\ \hline
$\CKMP$ & $\CK_g$ + $\MP_g$ \\ \hline
$\CKMP\ID$ & $\CKID_g$ + $\MP_g$ \\ \hline
\end{tabular}
\begin{tabular}{|l|l|} \hline
Hilbert System & Sequent System  \\ \hline \hline 
$\CKCEM$ & $\CKCEM_g$ \\ \hline
$\CKCEM\ID$ & $\CKCEM\ID_g$ \\ \hline
$\CKCEM\MP$ & $\CKCEM_g$ + $\MP_g$ + $\MPEM_g$ \\ \hline
$\CK\CEM\MP\ID$ & $\CKCEM\ID_g$ + $\MP_g$ + $\MPEM_g$  \\ \hline
\end{tabular}
\caption{Summary of Cut-Free Sequent Systems}
\label{fig:conditional-calculi}
\end{figure}

%
\subsection{Cut Elimination for Extensions of $\CK$ without $\CEM$}
We first treat extensions of the basic conditional logic $\CK$ with
axioms $\ID$ and $\MP$, but not including $\CEM$ and discuss $\CEM$
later, as the effect of adding $\CEM$ leads to a more general form
of the $\CK$ rule.

\noindent
If we absorb cuts using Lemmas \ref{lemma:adm-cut} and
\ref{lemma:adm-prop} we see that all
instances of 
\[ (\CK_g) \frac{A_0 = \dots = A_n \quad \neg B_1, \dots, \neg B_n,
B_0}{\neg (A_1 \cto B_1), \dots, \neg (A_n \cto B_n), (A_0 \cto
B_0), \Gamma}
\]
are admissible in $\Hilb\CK$. It is easy to see that the rule set
$\CK_g$ is actually absorbing:
\begin{thm} The rule set $\CK_g$ is absorbing and equivalent to
$\CK$. 
As a consequence, $\Gen\CK_g$ has cut-elimination and
$\Gen\CK_g \entails A$ iff $\Hilb\CK \entails A$ whenever $A \in
\Form(\Lambda)$.
\end{thm}
\begin{proof}
Using Lemmas \ref{lemma:adm-cut} and Lemma \ref{lemma:adm-prop} it is
immediate that the rule set $\CK_g$ is admissible in $\Hilb\CK$. The
argument that shows that $\CK_g$ is absorbing is analogous to that for
the modal logic $K$ (Example \ref{example:K}), and the result follows from Theorem
\ref{thm:gen-hilb-equiv}.
\end{proof}
%
%
%
\noindent
The logic $\CKID$ arises form $\CK$ by adding the identity axiom $A
\cto A$ to the rule set $\CK_H$ that axiomatises standard
conditional logic. Applying Lemma \ref{lemma:adm-cut} to the two rule
instances 
\[ \frac{A = A \quad \neg A, B}{\neg(A \cto A), (A \cto B)}
   \qquad
	 \frac{}{A \cto A}
\]
yields the (admissible) rule
\[ \frac{\neg A, B}{A \cto B}. \]
Again applying the same lemma, this time to a general instance of
$(\CK)$ and the rule that we just derived, that is,
\[ \frac{A_0 = C = A_1\dots = A_n \quad \neg D, \neg B_1, \dots,
\neg B_n, B_0}
   {\neg (C \to D), \neg(A_1 \cto B_1), \dots, \neg(A_n \cto B_n)}
	 \qquad
	 \frac{\neg C, D}
	 {C \cto D}
\]
now gives the rule
\[ (\CKID_g) \frac{A_0 = \dots = A_n \quad \neg A_0, \neg B_1,
\dots, \neg B_n, B_0}
   {\neg (A_1 \cto B_1), \dots, \neg(A_n \cto B_n), (A_0 \cto B_0),
	 \Gamma}
\] that can be again seen to be admissible by Lemma
\ref{lemma:adm-prop}.
It is easy to see that both $(\CK)$ and $(\ID)$ are derivable under
$(\CKID)$, and we note that $(\CKID_g)$ is admissible by construction.
If we denote the rule set
consisting of all instances of $\CKID$  by $\CK\ID_g$, we
obtain:
\begin{propn} \label{propn:ckid}
The rule set $\CKID_g$ is absorbing and equivalent to $\CK\ID$.
\end{propn}
\begin{proof}
It is easy to see that $\CKID_g$ absorbs the structural rules, and
that $\CK\ID$ is equivalent to $\CKID_g$. 

To see that $\CKID_g$ absorbs cut, we consider two instances of
$(\CKID_g)$, say
\[
  (r_1) \frac{A_0 = \dots = A_n \quad \neg B_1, \dots, \neg B_n,
	B_0}
	{\neg (A_1 \cto B_1), \dots, \neg(A_n \cto B_n), (A_0 \cto B_0),
	\Gamma}
\]
and
\[
  (r_2) \frac{C_0 = \dots = C_k \quad \neg D_1, \dots, \neg D_k,
	D_0}
	{\neg (C_1 \cto D_1), \dots, \neg(C_k \cto D_k), (C_0 \cto D_0),
	\Delta}
\]
and assume that the cut happens on $F \in \Form(\Lambda)$. The case
where $F \in \Gamma, \Delta$ is straightforward, so assume without
loss of generality that $F = (A_0 \cto B_0) = (C_1  \cto D_1)$. By
converting the equalities in the premise, and repeatedly applying
cut on $A_0 \equiv C_1$ we obtain
\[ \frac{C_0, \neg C_i \qquad A_0, \neg A_i}{C_0, \neg A_i}
\Cut(A_0 \equiv C_1)
\qquad
\frac{\neg C_0, C_i \qquad \neg A_0, A_i}{\neg C_0, A_i} \Cut(A_0
\equiv C_i)
\]
so that we obtain the derivability of
\[ \Sigma_1 = C_0 = A_1 = \dots = A_n = C_2 = \dots = C_k \]
in $\Gen\CK\ID_g + \Cut(F, r_1, r_2)$ (recall Notation \ref{notn:eq}). The derivation
\begin{prooftree}
\AxiomC{$\neg A_0, \neg B_1, \dots, \neg B_n, B_0$}
\AxiomC{$\neg C_0,  \neg D_1, \dots, \neg D_k, D_0$}
\RightLabel{$\Cut(B_0 \equiv D_1)$}
\BinaryInfC{$\neg A_0, \neg B_1, \dots, \neg B_n, \neg C_0, \neg
    D_2, \dots, \neg D_k, D_0$} 
\AxiomC{\!\!\!\!\!\!\!$\neg C_0, C_1$}
\RightLabel{$\Cut(C_1 \equiv A_0)$}
\BinaryInfC{$\neg C_0, \neg C_0, \neg B_1, \dots, \neg B_n, \neg
   D_2, \dots, \neg D_k, D_0$}
\UnaryInfC{$\neg C_0, \neg B_1, \dots, \neg B_n, \neg
   D_2, \dots, \neg D_k, D_0$}
\end{prooftree}
where contraction on $(C_0)$ was applied in the last step, shows
that 
\[ \Sigma_2 = \neg C_0, \neg B_1, \dots, \neg B_n, \neg
D_2, \dots, \neg D_k, D_0 \]
is derivable in $\Gen\CKID_g+ \Cut(F, r_1, r_2)$, and applying
$\CKID_g$ to $\Sigma_1$ and $\Sigma_2$ gives (cut-free) derivability of 
the desired sequent $\neg(A_1 \cto B_1), \dots, \neg(A_n \cto B_n),
\neg (C_2 \cto D_2), \dots, \neg (C_k \cto D_k), (C_0 \cto D_0),
\Gamma, \Delta$.
This completes the case distinction on $F$ and hence the proof of
the proposition.
\end{proof}
\noindent
Before we move to the next system, we briefly demonstrate the
derivation of the identity axiom in $\CKID_g$.
\begin{ex}
It is easy to say that $\Gen\CKID_g \entails A \cto A$ for all $A
\in \Form(\Lambda)$: we pick $n = 0$ to obtain the following
instance of $\CKID_g$
\[ \frac{\neg A, A}{A \cto A} \]
and note that the premise is in fact an axiom.
\end{ex}
%
%
The logic $\CKMP$ arises by augmenting the logic $\CK$ with the
additional axiom $(A \cto B) \to (A \to B)$. We briefly sketch the
construction of the additional axiom that gives rise to the rule
$(\MP_g)$ that we will use to establish cut-free completeness.

We consider a cut between an instance of $(\CK_g)$ and $(\MP)$, that
is, we have the derivation
\begin{prooftree}
\AxiomC{$A_0 = A_1 \qquad \neg B_1, B_0$}
\UnaryInfC{$\neg(A_1 \cto B_1), (A_0 \cto B_0$)}
\AxiomC{$\neg(A_0 \cto B_0), A_0 \to B_0$}
  \RightLabel{$\Cut(A_0 \cto B_0)$}
\BinaryInfC{$\neg (A_1 \cto  B_1), A_0 \to B_0$}
\end{prooftree}
that leads to the rule
\[ \frac{A_0 = A_1 \qquad \neg B_1, B_0 \qquad C = (A_0 \to B_0)}
   {(A_1 \cto B_1), C}
\]
by putting $C = A_0 \to B_0$. By Lemma \ref{lemma:adm-cut}, this
rule is admissible, and by Lemma \ref{lemma:adm-prop} so is the rule
\[
  \frac{A_1, C \qquad \neg B_1, C}{\neg(A_1 \cto B_1), C}
\]
and absorbing the structural rules (in particular contraction on $A
\cto B$ and inversion) leads to the general form
\[
  (\MP_g) 
	\frac{A, \neg(A \cto B), \Gamma \qquad \neg B, \neg (A \cto B),
	\Gamma}{\neg (A \cto B), \Gamma}
\] where we have elided the subscripts.
The effect of adding $(\MP)$ is similar to that of enriching the
modal logic $\K$ with the $(\T)$-axiom. 
We denote the rule set consisting of all instances of $\CK_g$
and $\MP_g$ by $\CKMP_g$. Our cut elimination theorem then takes the
following form:
\begin{propn} \label{thm:ckmp}
The rule set $\CKMP_g$ is absorbing and equivalent to $\CK\MP$.
\end{propn}

\begin{proof}
It is clear that both $(\CK_g)$ and $(\MP_g)$ absorb the structural
rules. For cut, we first 
consider cuts between two instances of $(\MP_g)$, say
\[
  (r_1) \frac{A, \neg(A \cto B), \Gamma \quad \neg B, \neg (A \cto
	B), \Gamma}{\neg (A \cto B), \Gamma}
	\qquad
  (r_2) \frac{C, \neg (C \cto D), \Delta \quad \neg D, \neg(C \cto
	D), \Delta}{\neg (C \cto D), \Delta}
\]
where the cut happens on $F \in \Form(\Lambda)$. We distinguish
several cases:

\emph{Case $F = (A \cto B)$ and $F \in \Delta$.} Then $\Delta = (A
\cto B), \Delta'$ for some $\Delta' \in \Seq(\Lambda)$. To eliminate the cut
on $C$, we note that the following two derivations 
\begin{prooftree}
\AxiomC{$A, \neg (A \cto B), \Gamma \quad \neg B, \neg (A \cto B), \Gamma$}
\RightLabel{$(\MP)$}
\UnaryInfC{$\neg (A \cto B), \Gamma$} \AxiomC{$C, \neg (C \cto D), (A \cto B), \Delta'$}
\RightLabel{(cut ($F$))}
\BinaryInfC{$C, \neg (C \cto D), \Gamma, \Delta'$}
\end{prooftree}
and
\begin{prooftree}
\AxiomC{$A, \neg (A \cto B), \Gamma \quad \neg B, \neg (A \cto B), \Gamma$}
\RightLabel{$(\MP)$)}
\UnaryInfC{$\neg (A \cto B), \Gamma$} \AxiomC{$\neg D, \neg (C \cto D), (A \cto B), \Delta'$}
\RightLabel{(cut ($F$))}
\BinaryInfC{$\neg D, \neg (C \cto D), \Gamma, \Delta'$}
\end{prooftree}
witness that we can use both $C, \neg (C \cto D), \Gamma, \Delta'$
and $\neg D, \neg (C \cto D), \Gamma, \Delta'$ as axioms in $\CKMP_g
+ \Cut(F, r_1, r_2)$ as the cuts occur between the premises of
$(r_2)$ and conclusions of $(r_1)$. Applying $(\MP_g)$ to these
axioms, we obtain that $\CKMP_g + \Cut (F, r_1, r_2) \entails \neg
(C \cto D), \Gamma, \Delta'$.

\emph{Case $F = (C \cto D)$ and $F \in \Gamma$.} This is symmetric to
the case above.

\emph{Case $F \in \Gamma$ and $\neg F \in \Delta$.} Then $\Gamma =
\Gamma', F$ and $\Delta = \Delta', \neg F$. We have to show that
\[ \neg (A \cto B), \neg (C \cto D), \Gamma', \Delta' \]
is derivable in $\Cut(F, r_1, r_2)$. We note that the deduction
\begin{prooftree}
\AxiomC{$A, \neg (A \cto B), F, \Gamma'$} 
\AxiomC{$C, \neg(C \cto D), \neg F, \Delta'$} 
\RightLabel{$(\Cut(F))$}
\BinaryInfC{$A, \neg (A \cto B), C, \neg (C \cto D), \Gamma', \Delta'$}
\end{prooftree}
witnesses that we may use $A, \neg (A \cto B), C, \neg (C \cto D),
\Gamma', \Delta'$ as an axiom in the system $\CKMP_g + \Cut(F, r_1,
r_2)$ as the cut on $F$ has occurred between premises of $r_1$ and
$r_2$. The same deduction, with $C$ replaced by $\neg D$ throughout,
witnesses that this is also the case for 
$A, \neg (A \cto B), \neg D, \neg (C \cto D), \Gamma', \Delta'$.
An application of $(\MP_g)$ now yields $\CKMP_g + \Cut(F, r_1, r_2)
\entails  \neg (C \cto D), A, \neg (A \cto B), \Gamma', \Delta'$.

By the symmetric argument (just replace $A$ by $\neg B$) we obtain that
also  $\CKMP_g + \Cut(F, r_1, r_2)
\entails  \neg (C \cto D), \neg B, \neg (A \cto B), \Gamma',
\Delta'$ and  an application of $(\MP_g)$ now yields
 $\CKMP_g + \Cut(F, r_1, r_2)
 \entails  \neg (C \cto D), \neg (A \cto B), \Gamma', \Delta'$
 as required.
%
%

What is left is to consider cuts, say on $F \in \Form(\Lambda)$,  between the conclusions of the
rules
\begin{align*}  (r_1) & \frac{A_0 = \dots = A_n \quad \neg B_1, \dots, \neg B_n,
B_0}{\neg (A_1 \cto B_1), \dots, \neg (A_n \cto B_n), (A_0 \cto
B_0), \Gamma} \\[2ex]
	(r_2) & \frac{C, \neg (C \cto D), \Delta \quad \neg D, \neg (C \cto
	D), \Delta}{\neg (C \cto D), \Delta}
\end{align*}
As before, we need to discuss several cases.

\emph{Case $F \in \Gamma$ or $\neg F \in \Gamma$.} Trivial, as the
conclusion of the cut can be derived using a different weakening
context $\Gamma$.

\emph{Case $F = (A_i \cto B_i)$ for some $1 \leq i \leq n$.} We assume without
loss of generality that $i = 1$ and have that $F  = (A_1 \cto B_1)
\in
\Delta$ so that $\Delta  = \Delta', F$. To replace the cut on $F$,
we consider the deduction 
{\small
\begin{prooftree}
\AxiomC{$A_0 = \dots = A_n \quad \neg B_1, \dots, \neg B_n, B_0$}
\RightLabel{(CK)}
\UnaryInfC{$\neg (A_1 \cto B_1), \dots, \neg(A_n \cto B_n), (A_0
\cto B_0), \Gamma$}
\AxiomC{$C, \neg (C \cto D), (A_1 \cto B_1), \Delta'$}
\RightLabel{(cut ($F$))}
\BinaryInfC{$\neg (A_2 \cto B_2), \dots, \neg (A_n \cto B_n), (A_0
\cto B_0), C, \neg (C \cto D), \Gamma, \Delta'$}
\end{prooftree}}
\noindent
which witnesses that we may use
\[
  \Sigma_1 = \neg (A_2 \cto B_2), \dots, \neg (A_n \cto B_n), (A_0
\cto B_0), C, \neg (C \cto D), \Gamma, \Delta' \]
as an axiom in
$\Gen\CKMP_g + \Cut(F, r_1, r_2)$. The above deduction, with $C$
replaced by $\neg D$ throughout, witnesses that the same is true for
\[ \Sigma_2 =  \neg (A_2 \cto B_2), \dots, \neg (A_n \cto B_n), (A_0
\cto B_0), \neg D, \neg (C \cto D), \Gamma, \Delta' \]
and applying $(\MP_g)$ with premises $\Sigma_1$ and $\Sigma_2$
yields $\Gen\CKMP_g + \Cut(F, r_1, r_2) \entails
\neg (A_2 \cto B_2), \dots, \neg (A_n \cto B_n), (A_0 \cto B_0),
\neg (C \cto D), \Gamma, \Delta'$ as required.

\emph{Case $F = (A_0 \cto B_0) = (C \cto D)$.} We have to give a
derivation of $\neg (A_1 \cto B_1), \dots, \neg (A_n \cto B_n),
\Gamma, \Delta$ in $\Cut(F, r_1, r_2)$. The deduction
\begin{prooftree}
\AxiomC{$A_0 = \dots = A_n$} \AxiomC{$\neg B_1, \dots, \neg B_n, B_0$}
\RightLabel{(CK)}
\BinaryInfC{$\neg (A_1 \cto B_1), \dots, \neg (A_n \cto B_n), (A_0 \cto B_0), \Gamma$}
\AxiomC{$\neg B_0, \neg (A_0 \cto B_0), \Delta$}
\RightLabel{(Cut ($F$))}
\BinaryInfC{$\neg (A_1 \cto B_1), \dots, \neg (A_n \cto B_n), \neg B_0, \Gamma, \Delta$}
\end{prooftree}
witnesses that we may use 
\[ \Sigma_1 = \neg (A_1 \cto B_1), \dots, \neg (A_n \cto B_n), \neg
B_0, \Gamma, \Delta \]
as an axiom in $\Gen\CKMP_g + \Cut(F, r_1, r_2)$ as the cut on
$F$ occurs between a conclusion of $(r_1)$ and a premise of $(r_2)$.
The same derivation, with $\neg B_0$ replaced by $A_0$ shows that
the same is true for
\[ \Sigma_2 = \neg (A_1 \cto B_1), \dots, \neg (A_n \cto B_n), A_0,
\Gamma, \Delta. \]
We therefore have the two derivations
\begin{prooftree}
\AxiomC{$\Sigma_1$}
\AxiomC{$\neg B_1, \dots, \neg B_n, B_0$}
\RightLabel{(Cut ($B_0$))}
\BinaryInfC{$\neg(A_1 \cto B_1), \dots, \neg (A_n \cto B_n), \neg
B_1, \dots, \neg B_n, \Gamma, \Delta$}
\end{prooftree}
and 
\begin{prooftree}
\AxiomC{$\Sigma_2$}
\AxiomC{$\neg A_0, A_1$}
\RightLabel{(w)}
\UnaryInfC{$\neg A_0, A_1, B_2, \dots, B_n$}
\RightLabel{(Cut ($A_0$))}
\BinaryInfC{$\neg(A_1 \cto B_1), \dots, \neg (A_n \cto B_n), A_1,
B_2, \dots, B_n, \Gamma, \Delta$}
\end{prooftree}
in $\Gen\CKMP_g + \Cut(F, r_1, r_2)$. Applying $(\MP_g)$ to the
conclusions of both yields that
$\Gen\CKMP_g + \Cut(F, r_1, r_2) \entails \Sigma_3$ where

\[
  \Sigma_3 = \neg (A_1, \cto B_1), \dots, \neg (A_n \cto B_n), \neg B_2,
\dots, \neg B_n, \Gamma, \Delta \]
as $\Gen\CKMP_g + \Cut(F, r_1, r_2)$ contains the contraction rule.
We now iterate the same scheme, where we use weakening on a
successively smaller subset of $B_2, \dots, B_n$. First, we note
that 
\begin{prooftree}
\AxiomC{$\Sigma_2$}
\AxiomC{$\neg A_0, A_2$}
\RightLabel{(w)}
\UnaryInfC{$\neg A_0, A_2, \neg B_3, \dots, \neg B_n$}
\RightLabel{(Cut ($A_0$))}
\BinaryInfC{$\neg(A_1 \cto B_1), \dots, \neg (A_n \cto B_n), A_2,
\neg B_3, \dots, \neg B_n, \Gamma, \Delta$}
\end{prooftree}
is a derivation in $\Gen\CKMP_g + \Cut(F, r_1, r_2)$ and applying
$(\MP_g)$ to $\Sigma_3$ and its conclusion yields $\Gen\CK\MP_g +
\Cut(F, r_1, r_2) \entails \Sigma_4$ where
\[
  \Sigma_4 = \neg(A_1 \cto B_1), \dots, \neg(A_n \cto B_n), \neg
	B_3, \dots, \neg B_n
\]
Iterating this scheme, we finally obtain 
\[\Gen\CKMP_g + \Cut(F, r_1, r_2) \entails
\neg (A_1 \cto B_1), \dots, \neg (A_n \cto B_n),
\Gamma, \Delta.
\]
Note that weakening and cuts on formulas of size $< \size(F)$ is
admissible in $\Gen\CKMP_g + \Cut(F, r_1, r_2)$.
%
%
%
%

\emph{Case $F = (A_0 \cto B_0)$ and $\neg F  \in \Delta$.} We have that
$\Delta = \neg (A_0 \cto B_0) , \Delta'$ and  the deduction
{\small
\begin{prooftree}
\AxiomC{$A_0 = \dots = A_n$} \AxiomC{$\neg B_1, \dots, \neg B_n, B_0$}
\RightLabel{(CK)}
\BinaryInfC{$\neg (A_1 \cto B_1), \dots, \neg (A_n \cto B_n), (A_0
\cto B_0), \Gamma$}
\AxiomC{$C, \neg (C \cto D), \neg (A_0 \cto B_0), \Delta'$}
\RightLabel{$(\Cut (F))$}
\BinaryInfC{$\neg (A_1 \cto B_1), \dots, \neg (A_n \cto B_n), C, \neg (C \cto D), \Gamma, \Delta'$}
\end{prooftree}}
\noindent witnesses that we may use
\[ \Sigma_2 = \neg (A_1 \cto B_1), \dots, \neg (A_n \cto B_n), C,
\neg (C \cto D), \Gamma, \Delta' \]
as an axiom in $\Gen\CKMP_g + \Cut(F, r_1, r_2)$. The same
derivation, with $C$ replaced by $\neg D$, shows that the same is
true for
\[ \Sigma_2 = \neg (A_1 \cto B_1), \dots, \neg (A_n \cto B_n), \neg
D, \neg (C \cto D), \Gamma, \Delta' \]
and applying $(\MP_g)$ with premises $\Sigma_1$ and $\Sigma_2$
yields the claim $\Gen\CKMP_g + \Cut(F, r_1, r_2) \entails \neg (A_1
\cto B_1), \dots, \neg (A_n \cto B_n), \Gamma, \Delta'$.
%
%
This finishes our analysis of cuts that may arise between
conclusions of the $(\CK_g)$ and the $(\MP_g)$-rule, and hence the proof.
\end{proof}
\noindent
As an example, we give a derivation of $(\MP)$ in the system
$\CKMP_g$ (in fact, a single application of $(\MP_g)$ suffices).
\begin{ex}
If expressed solely in terms of $\land$ and $\neg$, conditional
modus ponens takes the form $\neg ( (A \cto B) \land A \land \neg
B)$. The following derivation establishes that $(\MP)$ is derivable
in the above form
\begin{prooftree}
\AxiomC{$A, \neg(A \cto B), \neg A, B$}
\AxiomC{$\neg B, \neg(A \cto B), \neg A, B$}
\RightLabel{$(\MP_g)$}
\BinaryInfC{$\neg (A \cto B), \neg A, B$}
\RightLabel{$(\neg\neg)$}
\UnaryInfC{$\neg (A \cto B), \neg A, \neg \neg B$}
\RightLabel{$(\neg\land)$}
\UnaryInfC{$\neg(A \cto B), \neg(A \land \neg B)$}
\RightLabel{$(\neg\land)$}
\UnaryInfC{$\neg( (A \cto B) \land A \land \neg B)$}
\end{prooftree}
so that $(\MP)$ is derivable in $\CKMP_g$.
\end{ex}
\noindent
We now consider the logic that arises by adding both
conditional modus ponens $(A \cto B) \to (A \to B)$ and the identity
axiom $A \cto A$ to the logic $\CK$. In line with our naming
conventions, this logic is called $\CK\MP\ID$. To obtain a cut-free
axiomatisation of this logic, we consider the rule set $\CKMP\ID_g$
containing all instances of $\CKID_g$ and $\MP_g$. A close
inspection of the proof of Proposition \ref{thm:ckmp} gives that
$\CKMP\ID_g$ is absorbing, and therefore cut-free complete.

\begin{propn} \label{propn:ckmpid}
The rule set $\CK\MP\ID_g$ is absorbing and equivalent to
$\CK\MP\ID$.
\end{propn}
\begin{proof}
We follow the same strategy (and consider the same cases) as in the
proof of Proposition \ref{thm:ckmp} where we note that the
conclusions of $\CK_g$ and $\CKID_g$ are identical, the only
difference being that in
premise displayed on the far right in $\CK_g$ and $\CKID_g$ there is
one additional (negative) literal in
where $\CKID_g$. The proof of
Proposition \ref{thm:ckmp} can now be repeated literally by
adding this extra literal to all instances of $\CK_g$, thus turning
every instance of $\CK_g$ in the proof of Proposition \ref{thm:ckmp}
into an instance of $\CKID_g$.
\end{proof}

\subsection{Cut Elimination for Extensions of $\CKCEM$}

To construct an absorbing rule set for conditional logic plus the
axiom \[ (\CEM) (A \cto B) \lor (A \cto \neg B) \]
we start from the admissible rule set for $\CK$ and close under cuts
that arise with $(\CEM)$. Repeated applications of Lemma
\ref{lemma:adm-cut} and Lemma \ref{lemma:adm-prop} lead to the rule
set 
\[ (\CK\CEM_g) \frac{A_0 = \dots = A_n \quad B_0, \dots, B_j, \neg B_{j+1}, \neg
B_n}{(A_0 \cto B_0), \dots, (A_j \cto B_j), \neg (A_{j+1} \cto
B_{j+1}), \dots, \neg (A_n \cto B_n), \Gamma} \]
for $0 \leq j \leq n$. 
\begin{propn} \label{thm:ckcem}
The rule set $\CKCEM_g$ is absorbing and equivalent to $\CK\CEM$.
\end{propn}
\begin{proof}
Again, it suffices to check that the rule set $\CKCEM_g$ is absorbing,
where the absorption of structural rules is clear. It therefore
suffices to treat instances of cuts between conclusions of rules of
$\CKCEM$. Owing to the form of the $\CKCEM_g$-rule, our argument is
very similar to that used for $\CK_g$. 
We consider the following two instances 
\begin{align*} (r_1) & \frac{A_0 = \dots = A_n \quad B_0, \dots, B_i, \neg B_{i+1},
\dots, \neg B_n}{(A_0 \cto B_0), \dots, (A_i \cto B_i), \neg
(A_{i+1} \cto B_{i+1}), \dots, \neg (A_n \cto B_n), \Gamma} \\[2ex]
(r_2) & \frac{C_0 = \dots = C_m \quad D_0, \dots, D_j, \neg D_{j+1},
\dots, \neg D_m}{(C_0 \cto D_0), \dots, (C_j \cto D_j), \neg
(C_{j+1} \cto D_{j+1}), \dots, \neg (C_m \cto D_m), \Delta}
\end{align*}
and assume that the conclusions permit an instance of cut on $F \in
\Form(\Lambda)$. As usual, we distinguish several cases, where the
cases $F \in \Gamma, F \in \Delta, \neg F \in \Gamma$ and $\neg F
\in \Delta$ are trivial.

\emph{Case $F = (A_k \cto B_k) = (C_l \cto D_l)$ for $k > i$ and $0
\leq l \leq j$.} 
Without loss of generality we assume that $k = n$ and $l =
0$ and get $A_n = C_0$ and $B_n = D_0$. Denote the sequent that
arises from applying cut on $F$ to the conclusions of $r_1$ and
$r_2$ by $\Sigma_0$ and notice that, using cuts on
$A_n \equiv C_0$, we have that
\[ \Sigma = A_0 = A_1 = \dots = A_{n-1} = C_1 = \dots = C_m \]
is derivable in $\Gen\CKCEM_g + \Cut(F, r_1, r_2)$.  
This feeds into the derivation
{\small\begin{prooftree}
\AxiomC{$B_0, \dots, B_i, \neg B_{i+1}, \dots, \neg B_n$}
\AxiomC{$D_0, \dots, D_j, \neg D_{j+1}, \dots, \neg D_m$}
\RightLabel{(Cut ($B_n$))}
\BinaryInfC{$B_0, \dots, B_i, D_1, \dots, D_j, \neg  B_{i+1}, \dots,
\neg B_{n-1}, \neg D_{j+1}, \dots, \neg D_m$}
\AxiomC{$\Sigma$}
\RightLabel{$(\CKCEM)$}
\BinaryInfC{$\Sigma_0$}
\end{prooftree}}
which establishes that
$\Cut(F, r_1, r_2) \entails \Sigma_0$ as desired. 

The case $F = (A_k \cto B_k) \equiv (B_l \cto D_l)$ for $k > i$
and $1 \leq l \leq j$ is symmetric, which finishes the proof.
\end{proof}
\noindent
As a consequence, cut elimination holds in $\CKCEM_g$.  We show of
the $(\CEM)$ can be derived before moving on to the next calculus.
\begin{ex}
If we spell out the abbreviatios of $\lor$ in terms of $\neg$ and
$\land$, the axiom of conditional excluded middle takes the form
$\neg ( \neg(A \cto B) \land \neg (A \cto \neg B))$. The following
derivation in $\CKCEM_g$ shows that this form of the axiom is
derivable.
\begin{prooftree}
\AxiomC{$A = A$}
\AxiomC{$B, \neg B$}
\RightLabel{$(\CKCEM_g)$}
\BinaryInfC{$(A \cto B), (A \cto \neg B)$}
\RightLabel{$(\neg\neg)$}
\UnaryInfC{$\neg\neg(A \cto B), (A \cto \neg B)$}
\RightLabel{$(\neg\neg)$}
\UnaryInfC{$\neg\neg(A \cto B), \neg\neg(A \cto \neg B)$}
\RightLabel{$(\neg\land)$}
\UnaryInfC{$\neg(\neg (A \cto B) \land \neg (C \cto \neg B)$}
\end{prooftree}
In this derivation, we have chosen $j = n = 1$ and when applying
$(\CKCEM_g)$ and  choosing $j = 0$ yields an instance of $(\CK_g)$.
\end{ex}
\noindent
We now
consider the extension of $\CK$ with both conditional excluded
middle $(A \cto B) \lor (A \cto \neg B)$ and the identity axiom $(A
\cto A)$ and denote the ensuing logic by $\CK\CEM\ID$. As in the
construction of the rule set $\CK\ID_g$, we construct a rule set by
applying Lemma \ref{lemma:adm-cut} and \ref{lemma:adm-prop} by
considering cuts between an instance of $(\CKCEM_g)$ (left) and a
rule arising from a cut between $(\CK_g)$ and the identity axiom
(right)
\[ (\CK\CEM_g) 
\frac{A_0 = \dots = A_n \quad B_0, \dots, B_j, \neg B_{j+1}, \neg
B_n}{(A_0 \cto B_0), \dots, (A_j \cto B_j), \neg (A_{j+1} \cto
B_{j+1}), \dots, \neg (A_n \cto B_n), \Gamma} \qquad
\frac{\neg C, D}{C \cto D} \]
(where $1 \leq j \leq n$) leads to the (admissible) rule schema
\[ (\CK\CEM\ID_g) 
\frac{A_0 = \dots = A_n \quad \neg A_0, B_0, \dots, B_j, \neg B_{j+1}, \neg
B_n}{(A_0 \cto B_0), \dots, (A_j \cto B_j), \neg (A_{j+1} \cto
B_{j+1}), \dots, \neg (A_n \cto B_n), \Gamma} 
\]
that provides a cut-free axiomatisation of $\CK\CEM\ID$, as we now
show.

\begin{propn}
The rule set $\CK\CEM\ID_g$ is equivalent to $\CK\CEM\ID$ and
absorbing.
\end{propn}

\begin{proof}
Just as in the proof of Proposition \ref{thm:ckcem} we consider cuts
between two instances of $(\CK\CEM\ID_g)$ and note that the premises
of $(\CK\CEM_g)$ and $(\CK\CEM\ID_g)$ only differ by a negative
literal that is added to the premise on the far right in
$(\CK\CEM\ID_g)$. We consider precisely the same cases as in the
proof of Proposition \ref{thm:ckcem}. Using the same notation, we
note that the equality $A_0 = \dots = A_n$ in particular gives $\neg
A_n, A_0$ as a premise, and we use an additional cut on $A_0$,
followed by an instance of contraction, immediately prior to the
application of $(\CK\CEM)$ (that we replace by an instance of
$(\CK\CEM\ID_g)$) to show absorption of cut.
\end{proof}

We now consider extending
$\CK$ with both
conditional modus ponens and conditional excluded middle, but have
to take care of the cuts arising between $\MP_g$ and $\CK\CEM_g$, which
leads to the new rule
\[ (\MPEM_g) \frac{A, (A \cto B), \Gamma \qquad B, (A \cto B),
\Gamma}{(A \cto B), \Gamma} \] 
that was obtained in the same way as $(\MP_g)$. If we
denote the extension of $\CKCEM_g$ with $\MP_g$ and $\MPEM_g$ 
by $\CK\CEM\MP_g$, we
obtain:
\begin{propn} \label{thm:ckcemmp}
$\CK\CEM\MP_g$ is absorbing and equivalent to
$\CK\CEM\MP$. 
\end{propn}
\begin{proof}
It is clear that the rule set $\CK\CEM\MP_g$ absorbs the structural
rules and it is easy to see that it is equivalent to $\CK\CEM\MP$. We
have to show that it absorbs cut.

Cuts between the conclusions of two instances of $\MP_g$ have
already been treated in the proof of Theorem \ref{thm:ckmp}, and the
proof translates verbatim to cuts between instances of $\MPEM_g$.
We consider cuts between two instances
\[ (r_1) \frac{A, (A \cto B), \Gamma \qquad B, (A \cto B), \Gamma}
  {(A \cto B), \Gamma} \qquad
 (r_2) \frac{C, \neg (C \cto D), \Delta \qquad \neg D, \neg(C \cto
D), \Delta}{\neg (C \cto D), \Delta}
\]
where the cut is performed on $F \in \Form(\Lambda)$, say. The cases
where either
$F \in \Gamma$ and $\neg F$ in $\Delta$  or $F \in \Delta$ and $\neg
F \in \Gamma$ are straightforward.

\emph{Case $F = (A \cto B) = (C \cto D)$.}
The derivation
\begin{prooftree}
\AxiomC{$A, (A \cto B), \Gamma$}
\AxiomC{$B, (A \cto B), \Gamma$}
\RightLabel{($\MPEM$)}
\BinaryInfC{$(A \cto B), \Gamma$}
\AxiomC{$\neg B, \neg (A \cto B), \Delta$}
\RightLabel{(Cut $(A \cto B)$)}
\BinaryInfC{$\neg B, \Gamma, \Delta$}
\end{prooftree}
witnesses that we may use $\Sigma_1 = \neg B, \Gamma, \Delta$ as an
axiom in $\Gen\CK\CEM\MP + \Cut(F, r_1, r_2)$. Similarly, the derivation
\begin{prooftree}
\AxiomC{$A, \neg (A \cto B), \Delta$}
\AxiomC{$\neg B, \neg (A \cto B), \Delta$}
\RightLabel{($\MP$)}
\BinaryInfC{$\neg (A \cto B), \Delta$}
\AxiomC{$B, (A \cto B), \Gamma$}
\RightLabel{(Cut $(A \cto B)$)}
\BinaryInfC{$B, \Gamma, \Delta$}
\end{prooftree}
shows that the same is true for $\Sigma_2 = B, \Gamma, \Delta$: note
that in both cases, the cut was performed between an axiom and a
conclusion of both rules. As $\size(B) < \size(A \cto B)$ we may now
use cut on $B$ to establish that $\Gen\CK\CEM\MP_g + \Cut(F, r_1,
r_2) \entails \Gamma, \Delta$.

\emph{Case $F = (A \cto B)$ and $\neg F \in \Delta$.} We have that
$\Delta = \neg (A \cto B), \Delta'$. The derivation
\begin{prooftree}
\AxiomC{$A, (A \cto B), \Gamma$}
\AxiomC{$B, (A \cto B), \Gamma$}
\RightLabel{$(\MPEM)$}
\BinaryInfC{$(A \cto B), \Gamma$}
\AxiomC{$C, \neg (C \cto D), \neg (A \cto B), \Delta'$}
\RightLabel{(Cut ($F$))}
\BinaryInfC{$C, \neg (C \cto D), \Gamma, \Delta'$}
\end{prooftree}
witnesses that we may use 
\[ \Sigma_1 = C, \neg (C \cto D), \Gamma, \Delta' \]
as an axiom in $\Gen\CK\CEM\MP + \Cut(F, r_1, r_2)$. The same
derivation, with $C$ replaced by $\neg D$ shows that the same is
true for 
\[ \Sigma_2 = \neg D, \neg (C \cto D), \Gamma, \Delta' \]
and an application of $\MP_g$ yields derivability of
$\neg (C \cto D), \Gamma, \Delta'$.

\emph{Case $F = (C \cto D)$ and $\neg F \in \Gamma$.} Analogous by
interchaning the role of $\MP_g$ and $\MPEM_g$.

This leaves to consider cuts between two instances of $\CKCEM_g$ and
$\MP_g$ and between $\CKCEM_g$ and $\MPEM_g$.  We first consider
the rules
\begin{align*}
(r_1) &  \frac{A, (A \cto B), \Gamma \qquad B, (A \cto B), \Gamma}
{(A \cto B), \Gamma} \\[2ex]
(r_2) & \frac{A_0 = \dots = A_n \quad B_0, \dots, B_j, \neg B_{j+1},
\dots, \neg
B_n}{A_0 \cto B_0, \dots, (A_j \cto B_j), \neg (A_{j+1} \cto
B_{j+1}), \dots, \neg (A_n \cto B_n), \Delta} \end{align*}

\medskip\noindent In this setting, all cases except the case $F = (A
\cto B) = (A_i \cto B_i)$ with $i > j$ are entirely analogous to those
considered in the proof of Theorem \ref{thm:ckmp} where applications
of $\CK_g$ need to be replaced by applications of $\CKCEM_g$. In
case $F = (A \cto B) = (A_i \cto B_i)$ with $i > j$ we assume without
loss of generality that $i = n$ and argue, as in the proof of
Theorem \ref{thm:ckmp}, that 
\[ \Sigma_1  = (A_0 \cto B_0), \dots, (A_j \cto B_j), \neg (A_{j+1}
\cto B_{j+1}), \dots, \neg (A_{n-1} \cto B_{n-1}), B_n, \Gamma, \Delta
\]
and 
\[ \Sigma_2  = (A_0 \cto B_0), \dots, (A_j \cto B_j), \neg (A_{j+1}
\cto B_{j+1}), \dots, \neg (A_{n-1} \cto B_{n-1}), A_n, \Gamma, \Delta
\]
{\sloppy
both are axioms of $\Gen\CK\CEM\MP_g + \Cut(F, r_1, r_2)$, leading to
deductions ending in, repectively,  $\Sigma_1, B_0, \dots, B_j, \neg B_{j+1},
\dots, \neg B_{n-1}$ and $\Sigma_2, \neg A_0, B_1, \dots, B_j,
\neg B_{j+1}, \dots \neg B_{n-1}$. An application of
$\MP_g$ now yields derivability of 
$(A_0 \cto B_0), \dots, (A_j \cto B_j), \neg (A_{j+1} \cto B_{j+1}),
\dots, \neg (A_{n-1} \cto B_{n-1}), B_1, \dots, B_j, \neg B_{j+1},
\dots, \neg B_{n-1}$. Iterating the same schema, where $\MPEM_g$
is used instead of $\MP_g$ to eliminate occurrences of $\neg B_i$
for $i > j$ finally yields that $(A_0 \cto B_0), \dots, (A_j \cto
B_j), \neg (A_{j+1} \cto B_{j+1}), \dots, \neg (A_{n-1} \cto
B_{n-1}), \Gamma, \Delta$ is derivable in $\Gen\CK\CEM\MP_g +
\Cut(F, r_1, r_2)$.
}

To see that cuts between conclusions of $\CKCEM_g$ and $\MPEM_g$ can
be eliminated, one uses the same reasoning as above, with $\MPEM_g$
and $\MP_g$ interchanged.
\end{proof}

We note that the latter theorem was left as an open problem for the
sequent system presented in 
\cite{OlivettiEA07}. To complete the treatment of conditional
logics, we now turn to the system $\CK\CEM\MP\ID$
that arises by extending $\CK$ with the axioms correspoinding to
conditional excluded middle, conditional modus ponens and identity.
It follows by construction that the rule set
$\CK\CEM\MP\ID_g$, that we take as containing all instances of 
$\CK\CEM\ID_g$, $(\MP_g)$ and $(\MPEM_g)$ induces a calculus that
is equivalent to $\CK\CEM\MP\ID$ and we just have to establish
absorption.

\begin{propn}
The rule set $\CK\CEM\MP\ID_g$ is equivalent to $\CK\CEM\MP\ID_g$
and absorbing.
\end{propn}
\begin{proof}
It is clear that $\CK\CEM\MP\ID_g$ absorbs the structural rules
(this was established before for each rule schema). To see that cut
is absorbed, we proceed as in the proof of Proposition
\ref{thm:ckcemmp} where we replace every occurrence of $(\CKCEM_g)$
by the corresponding instance of $(\CK\CEM\ID_g)$. The additional
literal in the rightmost premise of $(\CK\CEM\ID_g)$ is treated in
the same way as in the proof of Proposition \ref{propn:ckid}.
\end{proof}

In summary, we obtain the following results about extensions of the
conditional logic $\CK$.
\begin{thm} Suppose that $\Log$ is a combination of $\ID$, $\MP$,
$\CEM$. Then $\Gen\Log_g \entails A$ whenever
$\Hilb\Log \entails A$ for all $A \in  \Form(\Lambda)$. Moreover, cut
elimination holds in $\Gen\Log$.
\end{thm}
\noindent
The theorem follows, in each of the cases, from Theorem
\ref{thm:cut-elim} and Theorem \ref{thm:gen-hilb-equiv} together
with the fact that the
rule set $\Log$ and $\Log_g$ are equivalent and the latter is absorbing.

\section{Complexity of Proof Search}

It is comparatively straightforward to extract complexity bounds for
provability of the logics considered above by analysing the complexity
of proof search under suitable strategies in the cut-free sequent
systems obtained. Clearly, in those cases where all modal rules peel
off exactly one layer of modal operators, the depth of proofs is
polynomial in the nesting depth of modal operators in the target
formula, and therefore, proof search is in PSPACE under mild
assumptions on the branching width of
proofs~\cite{SchroderPattinson09,PattinsonSchroder10}. Besides
reproving Ladner's classical result for $K$~\cite{Ladner77}, we thus
have
\begin{thm}\label{thm:CK-PSPACE}
  Provability in $\CK$ and $\CKID$ is in PSPACE.
\end{thm}
This reproves known complexity bounds originally shown
in~\cite{OlivettiEA07} (alternative short proofs using coalgebraic
semantics are given in~\cite{SchroderPattinson08d}).  For $\CKCEM$,
the bound can be improved to coNP using dynamic programming in the
same style as in~\cite{Vardi89}. This concept has to be handled
carefully when dealing with coNP bounds, however, as in
nondeterministic programs we cannot actually pretend that during the
execution of stage $n$ we have the results of the stages up to $n-1$
stored in memory --- otherwise, we could, e.g., just negate these
results and arrive at proving NP=coNP. Rather, dynamic programming
should be regarded as a metaphor for merging identical computations on
a non-deterministic machine; in particular, we need to take care to
use results of previous stages only \emph{positively} (as done
in~\cite{Vardi89}).

The point in our decision procedure where these considerations become
relevant is that we will wish to apply rule $(\CK\CEM_g)$
deterministically to subsequents that are as large as possible; i.e.\
we are interested in collecting \emph{maximal} sets of conditional
literals with provably equivalent antecedents, where the latter
equivalences are supposed to have been computed in previous
stages. Here, the maximality condition carries the danger of negative
use of previous results. The solution to this problem lies in the
following key lemma.
\begin{lem}\label{lem:cem-provable}
  Let $\Gamma=(A_0\cto B_0),\dots,(A_j\cto B_j),A_{j+1}\cto
  B_{j+1},\dots,A_n\cto B_n$ be a sequent consisting of conditional
  literals. Then $\Gamma$ is provable in $\CKCEM_g$ iff for every
  decomposition of $\{0,\dots,n\}$ into disjoint sets $I_0,\dots,I_k$
  ($k\ge 0$), \emph{one of} the following conditions holds.
  \begin{enumerate}[\em(1)]
  \item There exists $l$ such that $\{B_i\mid i\in I_l, 0\le i\le
    j\}\cup\{\neg B_i\mid i\in I_l,j+1\le i\le n\}$ is provable.
  \item There exist $l\neq r$ and $i\in I_l,p\in I_r$ such that
    $A_i=A_p$ is provable.
\end{enumerate}
\end{lem}

\begin{proof}
  \emph{Only if:} Since $\Gamma$ consists of conditional literals, any
  proof of $\Gamma$ must end in an application of rule $\CK\CEM_g$. Thus,
  there exists $I\subseteq\{0,\dots,n\}$ such that $A_i=A_p$ is
  provable for all $i,p\in I$ and $\{B_i\mid i\in I,0\le i\le
  j\}\cup\{\neg B_i\mid i\in I,j+1\le i\le n\}$ is provable. Now let
  $I_o,\dots, I_k$ be as in the statement. Then we have the following
  two cases:
  \begin{enumerate}[(1)]
  \item There exists $l$ such that $I\subseteq I_l$. In this case, the
    first alternative of the claim holds.
  \item We have $i,p\in I$ and $r\neq l$ such that $i\in I_r$, $p\in
    I_l$. In this case, the second alternative of the claim holds.
  \end{enumerate}

  \noindent \emph{If:} Define an equivalence relation on $\{0,\dots,n\}$ by
  taking $i$ and $p$ to be equivalent if $A_i=A_p$ is provable, and
  let $I_1,\dots,I_k$ be the induced disjoint decomposition of
  $\{0,\dots,n\}$ into equivalence classes. By construction, this
  decomposition does not satisfy the second alternative of the claim,
  hence it satisfies the first, which implies that $\Gamma$ is
  provable by applying rule $\CK\CEM_g$.
\end{proof}
\noindent This lemma now enables us to prove the announced coNP upper
bound:
\begin{thm}\label{thm:ckcem-complexity}
  Provability in $\CKCEM$ and in $\CK\ID\CEM$ is in coNP.
\end{thm}
\begin{proof}
  Since some aspects of our algorithm are more easily understood in NP
  style, we prove that \emph{unprovability} of a sequent $\Gamma$ can
  be decided in NP. We use dynamic programming as in~\cite{Vardi89}:
  we proceed in stages; at stage $i$, we decide unprovability of all
  sequents of the form $A,\neg B$ where $A$ and $B$ are subformulas of
  $\Gamma$ with nesting depth of conditionals at most $i$. We perform
  such stages up to $i=m-1$, where $m$ is the maximal nesting depth of
  conditionals in $\Gamma$. In a further, final stage, we then check
  unprovability of $\Gamma$. As there are at most linearly many
  stages, it suffices to show that each stage can be performed in NP,
  and since there are at most quadratically many candidate sequents in
  each stage, it suffices that unprovability of a single candidate
  sequent can be checked in NP at each stage.

  To this end, observe that proofs may generally be normalised to
  proceed as follows: first apply the propositional rules as long as
  possible, thus decomposing target sequents into sequents over
  \emph{conditional literals}, i.e.\ literals of the form $A\cto B$ or
  $\neg(A\cto B)$, in the various branches of the proof, and only
  apply $(\CK\CEM_g)$ when no more propositional rules are applicable
  (since all propositional rules monotonically increase the set of
  conditional literals when moving from the conclusion to the
  premises, it is clear that their -- backwards -- application never
  obstructs a possible application of $(\CK\CEM_g)$). The existential
  branching that arises from the conjunction rule ($A\land B,\Delta$
  is unprovable if either $A,\Delta$ or $B,\Delta$ is unprovable) is
  handled non-deterministically. It is clear that one can apply only
  linearly many propositional rules in any given branch of the
  computation.

  The application of rule $(\CK\CEM_g)$ after exhaustion of the
  propositional rules is handled according to
  Lemma~\ref{lem:cem-provable}: to check that a sequent of the form
  $\Delta=A_0\cto B_0,\dots,A_j\cto B_j,\neg(A_{j+1}\cto
  B_{j+1}),\dots,\neg(A_n\cto B_n)$ is unprovable, we guess a disjoint
  decomposition $I_1,\dots,I_k$ of $\{0,\dots,n\}$ and check that it
  violates Conditions (1) and (2) in the Lemma; the negation of these
  conditions introduces universal quantifiers over polynomial-sized
  ranges, which we check deterministically. Here, checking violation
  of Condition (2) amounts to using unprovability of quadratically
  many sequents checked in previous stages, which from the perspective
  of the present stage can be done in polynomial time. Checking
  violation of Condition (1) is more problematic, as it involves a
  recursive check of unprovability of sequents $\Delta_l=\{B_i\mid
  i\in I_l, 0\le i\le j\}\cup\{\neg B_i\mid i\in I_l,j+1\le i\le n\}$
  for all $1\le l\le k$; specifically, we have to ensure that this
  recursion, whose depth is limited by the nesting depth of
  conditionals in $\Gamma$, does not lead to exponentially long
  computation paths.

  To this end, we note that the breadth of the part of the proof
  tree that we explore in one computation is given by a function
  $f(\Delta)$ that obeys a recursive equation of the form
  \begin{equation*}
    f(\Delta)=\sum_{l=0}^k f(\Delta_l)
  \end{equation*}
  where
  \begin{equation}\label{eq:breadth}
    \size(\Delta)\ge\sum_{l=0}^k\size(\Delta_l)
  \end{equation}
  because $I_0,\dots,I_k$ is a disjoint decomposition of
  $\{0,\dots,n\}$. It follows easily that $f(\Delta)$ is at most
  linear in $\Delta$, and hence the overall size of the part of the
  proof tree explored in one computation is at most quadratic. This
  finishes the proof for the case of $\CK\CEM$. The proof for
  $\CK\ID\CEM$ is entirely analogous, noting that although the main
  premise in rule $(\CK\ID\CEM_g)$ is by one literal $\neg A_0$ larger
  that in the case of $(\CK\CEM_g)$, estimate~(\ref{eq:breadth})
  remains true.
\end{proof}

\noindent
More interesting are those cases where some of the modal operators
from the conclusion remain in the premise, such as $\T$, $\K4$,
$\CKMP$, and $\CKMPCEM$ (where the difference between non-iterative
logics, i.e.\ ones whose Hilbert-axiomatisation does not use nested
modalities, such as $\T$, $\CKMP$, and $\CKMPCEM$, and iterative
logics such as $\K4$ is surprisingly hard to spot in the sequent
presentation). For $\K4$, the standard approach is to consider proofs
of minimal depth, which therefore never attempt to prove a sequent
repeatedly, and analyse the maximal depth that a branch of a proof can
have without repeating a sequent. For $\T$, a different strategy is
used, where the $(T)$ rule is limited to be applied at most once to
every formula of the form $\neg\Box A$ in between two applications of
$(K)$~\cite{HeuerdingEA96}. A similar strategy works for the
conditional logics $\CKMP$ and $\CKMPCEM$, which we explain in some
additional detail for $\CKMP$.

We let $\CKMP_g^0$ and $\CKMP^1_g$ denote restricted sequent systems,
defined as follows.
\begin{enumerate}[$\bullet$]
\item In $\CKMP_g^0$, a formula $\neg(A\cto B)$ is \emph{marked} on a branch
  as soon as the rule $(\MP_g)$ has been applied to it (backwards) and
  unmarked only at the next application of rule $(\CK_g)$. Rule
  $(\MP_g)$ applies only to unmarked formulas.
\item In $\CKMP_g^1$, we instead impose a similar restriction where
  rule $(\MP_g)$ applies to a sequent $\neg(A\cto B),\Gamma$ only in
  case $\Gamma$ does not contain a propositional descendant of either
  $A$ or $\neg B$. Here, a sequent $\Delta$ is called a
  \emph{propositional descendant} of a formula $A$ if it can be
  generated from $A$ by applying propositional sequent rules
  backwards. Formally, the set $\mathcal{D}$ of propositional
  descendants of $A$ is the closure of $\{A\}$ under the inversion
  rules.
  (E.g.\ the propositional descendants of $(\neg(A\land B)\land C)$
  are $(\neg(A\land B)\land C)$; $\neg(A\land B)$; $C$; and $\neg A,\neg B$.)
\end{enumerate}
Our goal is to show that $\CKMP_g$, $\CKMP^0_g$, and $\CKMP^1_g$ prove
the same sequents. Here, two inclusions are easy to show:
\begin{lem}
  Every sequent that is provable in $\CKMP^1_g$ is provable in
  $\CKMP^0_g$, and every sequent that is provable in $\CKMP^0_g$ is
  provable in $\CKMP_g$.
\end{lem}
\begin{proof}
  The second implication is trivial, since $\CKMP^0_g$ explicitly
  restricts $\CKMP_g$. The first implication follows from the fact
  that whenever an occurrence of a formula $\neg (A\cto B)$ is marked
  in a sequent $\neg(A\cto B),\Gamma$ in a $\CKMP^0_g$ proof, then
  $\Gamma$ must contain a propositional descendant of either $A$ or
  $\neg B$.
\end{proof}
\noindent Next, we observe:
\begin{lem}\label{lem:ckmp1-struct}
  The system $\CKMP_g^1$ admits inversion.
\end{lem}
\begin{proof}
  The inductive proof for $\CKMP_g$ can just be copied due to the fact
  that absorption of inversion by $\CKMP_g$ never involves the
  introduction of additional applications of $(\MP_g)$, and the
  conclusion of instances of inversion never introduces additional
  propositional descendants (unlike, e.g., in the case of weakening).
\end{proof}
\noindent This enables us to prove the missing inclusion:
\begin{lem}
  Every sequent that is provable in $\CKMP_g$ is provable in
  $\CKMP^1_g$.
\end{lem}
\begin{proof}
  By Lemma~\ref{lem:ckmp1-struct}, it suffices to prove that we can
  replace backwards applications of $(\MP_g)$ to sequents $\neg(A\cto
  B),\Gamma$ with $\Gamma$ containing a propositional descendant of
  either $A$ or $\neg B$, with subproofs using inversion. This is
  clear: e.g.\ if $\Gamma$ contains a propositional descendant of $A$,
  then $\neg(A\to B),\Gamma$ can be proved from $\neg(A\to
  B),A,\Gamma$ alone by repeated application of inversion.
\end{proof}
\begin{cor}\label{cor:ckmp-equiv}
  The systems $\CKMP_g$ and $\CKMP_g^0$ prove the same sequents.
\end{cor}
\noindent
This determines the complexity of proof search in $\CKMP$:
\begin{cor}
  Provability in $\CKMP$ is in PSPACE.
\end{cor}
\begin{proof}
  By Corollary~\ref{cor:ckmp-equiv}, it suffices to show that proof search
  in $\CKMP_g^0$ is in PSPACE. The latter is shown analogously to
  Theorem~\ref{thm:CK-PSPACE}, as proofs in $\CKMP_g^0$ are easily
  seen to have at most polynomial depth.
\end{proof}
\noindent The same line of reasoning applies essentially without
change to $\CK\ID\MP$, so that \emph{provability in $\CK\ID\MP$ is in
  PSPACE}.

The same approach works for logics that include $(\CEM)$, with the
only actual modification being that in the systems $\CK\MP\CEM_g^0$
and $\CK\ID\MP\CEM_g^0$, backwards application of both $(\MP_g)$ and
$(\MPEM_g)$ is restricted to unmarked formulas. Equivalence of the
restricted systems to the full systems is shown in the same manner as
for $\CKMP$. We then have, in analogy to
Theorem~\ref{thm:ckcem-complexity}
\begin{thm}\label{thm:ckmpcem-complexity}
  Provability in $\CKMPCEM$ and in $\CK\ID\MP\CEM$ is in PSPACE.
\end{thm}
\noindent We note that the complexity of $\CKMPCEM$ was explicitly
left open in~\cite{OlivettiEA07}. We also note that regrettably we
were not able to reproduce our claim from~\cite{SchroderPattinson09}
that $\CKMPCEM$ is in coNP, the problem being that the
estimate~(\ref{eq:breadth}) breaks down in the presence of $(\MP_g)$ and
$(\MP\CEM)_g$; since no better lower bound than coNP is currently
known for $\CKMPCEM$ and $\CK\ID\MP\CEM$, this means that the exact
complexity of these logics remains open.

\section{Conclusions}

\noindent We have established a generic method of cut elimination in
modal sequent system based on absorption of cut and structural rules
by sets of modal rules. We have applied this method in particular to
various conditional logics, thus obtaining cut-free unlabelled sequent
calculi that complement recently introduced labelled
calculi~\cite{OlivettiEA07}. In at least one case, the conditional
logic $\CKMPCEM$ with modus ponens and conditional excluded middle,
our calculus seems to be the first cut-free calculus in the
literature, as cut elimination for the corresponding calculus
in~\cite{OlivettiEA07} was explicitly left open. We have applied these
calculi to obtain complexity bounds on proof search in conditional
logics; in particular we have reproved known upper complexity bounds
for $\CK$, $\CKID$, $\CKMP$~\cite{OlivettiEA07} and improved the bound
for $\CKCEM$ and $\CK\ID\CEM$ from PSPACE to coNP using dynamic
programming techniques following~\cite{Vardi89}. Moreover, we have
obtained an upper bound PSPACE for $\CKMPCEM$, for which no bound has
previously been published; a strong suspicion remains, however, that
this logic is actually in coNP.  We conjecture that our general method
can also be applied to other base logics, e.g. intuitionistic
propositional logic or first-order logic; this is the subject of
further investigations.



\providecommand{\noopsort}[1]{}


\end{document}